\documentclass[letterpaper]{article} 
\usepackage{aaai2027}  

\usepackage[hyphens]{url}  
\usepackage{graphicx} 
\usepackage{natbib}  
\usepackage{caption} 
\usepackage{algorithm}

\usepackage{newfloat}
\usepackage{listings}
\DeclareCaptionStyle{ruled}{labelfont=normalfont,labelsep=colon,strut=off} 
\floatstyle{ruled}
\newfloat{listing}{tb}{lst}{}
\floatname{listing}{Listing}

\usepackage{booktabs}

\usepackage[noend]{algpseudocode}

\algdef{SE}[DOWHILE]{Do}{doWhile}{\algorithmicdo}[1]{\algorithmicwhile\ #1}%

\usepackage{graphicx}
\usepackage{amsmath}
\usepackage{amsthm}
\usepackage{amssymb}
\usepackage{thmtools, thm-restate}
\usepackage{tikz}
\usetikzlibrary{calc,shapes,arrows}
\usepackage{pgfplots}
\pgfplotsset{compat=1.18}
\usepackage{xspace}
\usepackage{paralist}
\usepackage{mathtools}
\usepackage{todonotes}

\title{Progression- vs Automata-based Anticipatory Monitoring of LTL over Finite Traces (Extended Version)}
\author{
Sarah Winkler\textsuperscript{\rm 1},
Toryn Klassen\textsuperscript{\rm 2},
Sheila McIlraith\textsuperscript{\rm 2},
Marco Montali\textsuperscript{\rm 1}
}
\affiliations {
    \textsuperscript{\rm 1}Free University of Bozen-Bolzano, Italy\\
    \textsuperscript{\rm 2}University of Toronto, Canada\\
\{winkler, montali\}@inf.unibz.it,
\{toryn, sheila\}@cs.toronto.edu
}

\newtheorem{theorem}{Theorem}

\newtheorem{proposition}[theorem]{Proposition}
\newtheorem{remark}[theorem]{Remark}
\newtheorem{definition}{Definition}
\newtheorem{example}{Example}

\renewcommand{\phi}{\varphi}

\newcommand{\U}{\mathrel{\mathsf{U}}}
\newcommand{\R}{\mathrel{\mathsf{R}}}
\newcommand{\G}{\mathsf{G}\xspace}
\newcommand{\F}{\mathsf{F}\xspace}
\newcommand{\X}{\mathsf{X}\xspace}
\newcommand{\wX}{\mathsf{X}_w\xspace}

\newcommand{\mc}[1]{\mathcal{#1}}

\newcommand{\ints}{\mathit{int}}
\newcommand{\rats}{\mathit{rat}}
\newcommand{\bools}{\mathit{bool}}
\newcommand{\last}{\mathit{last}}
\newcommand{\CC}{\mc C}
\renewcommand{\AA}{\mc A}
\newcommand{\LL}{\mathcal L}

\newcommand{\DFA}[1][\psi]{\mathcal D_{#1}}

\newcommand{\inn}{\,{\in}\,} 

\newcommand{\pre}[1]{{#1}^-} 
\newcommand{\trace}[1]{\langle #1\rangle}

\newcommand{\m}[1]{\mathsf{#1}}

\newcommand{\dom}{\mathit{dom}}

\newcommand{\lemref}[1]{Lem.~\ref{lem:#1}}

\newcommand{\defref}[1]{Def.~\ref{def:#1}}

\newcommand{\thmref}[1]{Thm.~\ref{thm:#1}}

\newcommand{\exaref}[1]{Ex.~\ref{exa:#1}}
\newcommand{\figref}[1]{Fig.~\ref{fig:#1}}

\newcommand{\LTLf}{\textsc{LTL}$_f$\xspace}
\newcommand{\LTL}{\textsc{LTL}\xspace}
\newcommand{\Declare}{\textsc{Declare}\xspace}

\newcommand{\aLTLf}{\textsc{aLTL}$_f$\xspace}

\newcommand{\PS}{\text{\textsc{ps}}}
\newcommand{\PV}{\textsc{pv}}
\newcommand{\CS}{\textsc{cs}}
\newcommand{\CV}{\textsc{cv}}

\newcommand{\llbracket}{[\![}
\newcommand{\rrbracket}{]\!]}
\newcommand{\RV}{\mathit{RV}}
\newcommand{\progress}{\mathit{progress}}

\newcommand{\sarahtodo}[2][]{\todo[linecolor=yellow!80!red,backgroundcolor=yellow!80!red!25,bordercolor=yellow!80!red, #1]{\tiny{Sarah: #2}}}

\definecolor{darkgreen}{rgb}{0.2, 0.9, 0.6}

\newcommand{\progmon}{\textsc{ProgMon}\xspace}
\newcommand{\dfamon}{\textsc{DfaMon}\xspace}

\begin{document}


\newcommand{\rephrase}[1]{\textcolor{olive}{#1}}
\newcommand{\revise}[1]{\textcolor{blue}{#1}}
\newcommand{\review}[1]{\textcolor{blue}{#1}}

\newcommand{\alt}[1]{\textcolor{brown}{#1}}
\newcommand{\althide}[1]{}

\newcommand{\added}[1]{\textcolor{red}{#1}}
\newcommand{\todoo}[1]{\textcolor{blue}{({\bf TO DO:} #1)}}

\newcommand{\myhide}[1]{}

\newcommand{\remove}[1]{\textcolor{green}{#1}}
\newcommand{\removemaybe}[1]{\textcolor{orange}{#1}}
\newcommand{\removecrcr}[1]{}
\newcommand{\removecrc}[1]{\textcolor{orange}{#1}}
\newcommand{\removehide}[1]{}

\newcommand{\removeifneeded}[1]{\textcolor{cyan}{#1}}

\newcommand{\commentsmKEEP}[1]{\textcolor{cyan}{({\bf SM:} #1)}}

\newcommand{\draftnote}[1]{\textcolor{blue}{#1}}

\newcommand{\drafthide}[1]{}

\newcommand{\toaddcr}[1]{}
\newcommand{\toaddcrc}[1]{}
\newcommand{\addCRC}[1]{}
\newcommand{\addtocr}[1]{}



\newif\ifcomments
\commentsfalse

\newcommand{\commentsmbox}[2][]{\generalCommentFlex[#1]{cyan}{SM}{#2}}
\newcommand{\commenttkbox}[2][]{\generalCommentFlex[#1]{magenta}{TK}{#2}}
\newcommand{\commentswbox}[2][]{\generalCommentFlex[#1]{violet}{SW}{#2}}
\newcommand{\commentmmbox}[2][]{\generalCommentFlex[#1]{orange}{MM}{#2}}

\ifcomments

\setlength\marginparwidth{15mm}
\newcommand{\todocustom}[3]{\todo[linecolor=#2,backgroundcolor=#2!25,bordercolor=#2,#3]{#1}}
\newcommand{\generalCommentFlex}[4][]{\todocustom{{\bf #3}: #4}{#2}{inline,size=\scriptsize,caption={},#1}}

\newcommand{\commentsm}[1]{\textcolor{cyan}{({\bf SM:} #1)}}

\newcommand{\commenttk}[1]{\textcolor{magenta}{({\bf TK:} #1)}}

\newcommand{\commentsw}[1]{\textcolor{violet}{({\bf SW:} #1)}}

\newcommand{\commentmm}[1]{\textcolor{magenta}{({\bf MM:} #1)}}

\newcommand{\revisit}[1]{\textcolor{blue}{#1}}

\else

\newcommand{\commentsm}[1]{}
\newcommand{\commentsmhide}[1]{}
\newcommand{\commenttk}[1]{}
\newcommand{\commentsw}[1]{}
\newcommand{\commentmm}[1]{}

\renewcommand{\sarahtodo}[1]{}

\renewcommand{\added}[1]{#1}
\newcommand{\revisit}[1]{#1}
\renewcommand{\removeifneeded}[1]{#1}

\newcommand{\generalCommentFlex}[4][]{}

\fi


\maketitle

\begin{abstract}
When safety-critical systems are developed from a known internal
specification, their correctness can be established by model checking.
In the frequent case where such a specification is 
unknown or inaccessible, runtime verification 
presents an attractive alternative, e.g., to ascertain that
autonomous and agentic systems as well as business processes satisfy desirable properties and/or comply with safety
requirements.
In this paper we study \emph{anticipatory monitoring}, \removeifneeded{an advanced form of runtime verification,} where the monitoring state is determined by both the trace prefix seen so far, and all its possible finite-length, future continuations. We focus on monitoring
linear-time properties that may involve arithmetic constraints. Automata-based approaches, the de-facto standard in this setting, are notorious for their computational complexity. We propose an alternative approach based on progression
and \LTLf satisfiability checking, for both propositional and arithmetic settings. We experimentally compare the automata- and progression-based approaches, and a third method that combines the two. Our experiments suggest that the progression-based approach often succeeds in producing a verdict when the automata constructions do not terminate, especially for the arithmetic setting. For the propositional setting, the combined technique provides a good tradeoff.

\end{abstract}

\section{Introduction}
\label{sec:intro}

Checking whether a dynamical system operates as intended is \added{central to realizing} reliable, trustworthy
artificial intelligence (AI) and other systems. 
When the system under scrutiny comes with a fully accessible specification of its internal functioning, this can be accomplished through model-based formal analysis, such as model checking \cite{BK08}. Unfortunately, most modern systems defy such formal analysis because their internal specification is inaccessible, too complex to formalize, or \added{because} subsystems operate as black boxes.
Notable examples include:
\begin{inparaenum}[\itshape (i)]
\item safety-critical autonomous settings \cite{BCCG25} such as robotic \cite{tessla-ros}, aerial \cite{MoRS17}, and cyber-physical \cite{NCKA25} systems;
\item business processes with compliance requirements \cite{LMMR15,DDMM22};
\item \added{properties relating to} agentic systems, ranging from interaction protocols \cite{ChMMT13,AJCD16,DaTY18} to \added{attainment of} rewards~\cite{CamachoIKVM19,GiacomoFIPR20,AdalatB26}, 
and overall alignment \cite{CCGD26} including \added{of }LLMs \cite{AlKM26}.
\end{inparaenum}

In such settings, an attractive solution is to pair the system with a monitor that tracks the 
events occurring during an execution, and verifies at runtime whether the system 
\revisit{is operating}
as expected.
\addCRC{In the current software ecosystem where
AI coding agents and web apps enable rapid development
and distribution of software with few assurances of software
quality, such monitoring can provide some 
quality guarantees.} 
\emph{Runtime verification} (RV)
\added{is an established}
approach to that end: given a property of interest $\varphi$ expressed in a suitable logic, a corresponding monitor is automatically synthesized to check whether an  execution conforms to $\varphi$.
In this work we adopt linear-time temporal logic over finite traces (\LTLf) \cite{deGiacomoV13}  as the specification logic, possibly extended with arithmetic constraints (\aLTLf) as in \cite{FMPW23}.
%
We consider \emph{anticipatory monitoring} \cite{BaLS10} where monitors not only check whether the trace $\tau$ observed so far satisfies the property of interest, but also how possible continuations of $\tau$ affect its satisfaction.
More precisely, in the RV-LTL  semantics,  after each trace $\tau$ the input property $\varphi$  can be either permanently satisfied ($\varphi$ is satisfied by $\tau$ and will stay so no matter how $\tau$ is continued), or currently satisfied ($\varphi$ is satisfied by $\tau$, but there is a continuation of $\tau$ that can violate it) or have one of the dual values of permanent and current violation.

\myhide{
\begin{example}
Consider a simple purchase-to-pay process where whenever an order is closed it must get paid, $\G( \m{close} \to \F \m{pay})$ in \LTLf \added{(See Preliminaries)}; and whenever an order is cancelled, it may not be paid: $\G( \m{cancel} \to \neg\F \m{pay})$. Given a trace $\trace{\{\m{close}\}, \{\m{cancel}\}}$, the conjunction of these two formulas is \emph{permanently violated} as no continuation can satisfy the contradictory requirements. 
\end{example}
}

The de-facto standard technique to perform RV-LTL monitoring of \LTL and \LTLf is based on deterministic finite automata (DFA) \cite{BaLS11,MMWV11,DDMM22}. This method is \emph{eager} in that  a DFA is constructed during a preprocessing phase as a trace-independent monitoring device that is subsequently used to monitor one or more traces. This approach has been extended to \aLTLf~\cite{FMPW23}, but suffers in either setting from the doubly-exponential preprocessing complexity. Even more critical, since anticipatory monitoring of \aLTLf is not solvable in general, for \aLTLf this method often fails to terminate in the preprocessing phase.

In this paper, we introduce a complementary, \emph{lazy} anticipatory monitoring technique that relies on \emph{progression}~\cite{BacchusK00} and \LTLf satisfiability checking.
Our approach 
differs from automata-based monitoring in two respects: 
\begin{inparaenum}[\itshape (i)]
\item it does not incur any preprocessing cost, but distributes it incrementally at runtime; 
\item it does not construct a general monitoring device valid for all traces, but contextualises and localises reasoning to the trace under scrutiny.
\end{inparaenum}
Notably, we show that essentially the same progression-based approach applies to both propositional \LTLf and the more sophisticated, data-aware setting of \aLTLf. 
Our goal is then to compare, relate and combine the two complementary approaches of automata- and progression-based monitoring. 

More precisely, our contributions are:
\begin{inparaenum}[\itshape (i)]
\item an anticipatory \LTLf monitoring technique based on progression and satisfiability checking (\progmon), which we rigorously prove to be correct;
\item an implementation of \progmon, as well as a combined approach that applies DFA-based monitoring (\dfamon) and \progmon in parallel;
\item 
an extensive experimental evaluation comparing \progmon, \dfamon and the combined approach on different data sets, including real-life event logs from process mining;
\item the lifting of \progmon to the data-aware temporal logic \aLTLf, based on a suitable definition of progression for \aLTLf;
\item a demonstration that \progmon terminates on relevant fragments of \aLTLf, including new solvability results; and
\item the extension of the comparative evaluation to \aLTLf. 
\end{inparaenum}
Our main results are that for the propositional case, \progmon can handle 18\% more and the combined approach about 30\% more benchmarks than \dfamon because the latter experiences timeouts on large formulae, though for cases where both succeed \dfamon is about twice as fast. The combined approach\added{, being faster than progression,} offers a good tradeoff. 
The lazy and trace-tailored approach of \progmon is particularly appealing for the arithmetic setting, where DFA-based monitoring often fails to terminate in the preprocessing phase: \progmon is both faster and handles more benchmarks than \dfamon.

In the next 
section we provide further background and related work, followed by a section with technical preliminaries. After that, we define \progmon first for the propositional, and then for the arithmetic case. We present our evaluation before we  conclude by discussing the implications of our findings.
Proofs, further experimental data, and code can be found in the supplementary material.

\section{Background and Related Work}
Most monitoring approaches in the literature limit themselves to checking whether the trace prefix seen so far satisfies the desired property~\cite{DAngeloSSRFSMM05,BartocciFFR18,BasinKMZ15}. Anticipatory monitoring guarantees instead that violations are detected at the earliest possible moment \cite{MMWV11,DDMM22}, enabling to promptly react to issues at runtime, which is akin to process enactment \cite{PFPA13,DDMM22} and agent regimentation \cite{EMST03}. However, this task is more complex as it involves temporal reasoning as opposed to mere query answering.
While RV-LTL, also known as LTL$_4$~\cite{MedhatBFJ26}, has been originally defined for LTL over infinite traces, it has later been widely investigated in the finite-trace setting \cite{MMWV11,DDMM22}, considering properties based on \LTLf and extensions thereof \cite{deGiacomoV13}. Interestingly, while not all LTL properties are monitorable in RV-LTL \cite{BaLS11}, every \LTLf formula actually is. 

Automata-based monitoring techniques \cite{BaLS11,MMWV11,DDMM22} for RV-LTL construct a  DFA for the input property, that is directly turned into a monitor by assigning each DFA state to one of the four RV-LTL truth values through simple reachability checks. 
This approach has been lifted to data-aware variants of \LTLf, combining such automata-theoretic techniques with SMT reasoning \cite{FMPW23}. 
However, while tremendous algorithmic advancements have been achieved in building and manipulating DFAs for \LTLf (see, e.g., \cite{BansalLTV20,XiaoL0SPV21,DeGF21,DuretLutzZPGV25}), automata construction remains a bottleneck, and makes monitoring hardly feasible when the input property $\phi$ is a complex formula. This is the standard case in declarative business process monitoring \cite{LMMR15,CicM22}, where  $\phi$ is typically a large conjunction of process constraints.

Progression is at the core of some of the most advanced techniques for \LTLf automata construction, satisfiability, and synthesis \cite{GiacomoFLVX022,NXXX23,XLZS26}. In the context of runtime verification it has been so far investigated for forms of prefix-monitoring without anticipation \cite{RosuH05,LengH19} and also for anticipatory monitoring but with a different finite-trace semantics~\cite{BauerFMSD2016}. \addCRC{Consider adding Alamdari et al.}

\section{Preliminaries}
\label{sec:background}

We first recall propositional linear-time temporal logic over finite traces (\LTLf)~\cite{deGiacomoV13}:

\begin{definition}
\label{def:language}
For a set of propositions $P$ and $p\in P$, the set of properties  $\LL_P$ is defined by the following grammar:
\[\psi :: = 
p \mid 
\neg \psi \mid 
\psi \wedge \psi \mid 
\psi \vee \psi \mid
\X \psi \mid 
\wX \psi \mid
\psi \U \psi\mid
\psi \R \psi.
\]
\end{definition}

$\wX$ is the weak next operator, and
the usual transformations apply, namely 
$\bot \equiv (p \wedge \neg p)$ for an arbitrary $p\,{\in}\,P$, $\top \equiv \neg \bot$,
$\F \psi \equiv \top \U \psi$, and
$\G \psi \equiv \bot \R \psi$ \added{where $\F$,$\G$,$\U$, and $\R$ denote Finally (eventually), Globally (always), Until, and Release}.

A \emph{trace} is a sequence $\tau \in (2^P)^*$ written as 
$\tau = \trace{w_0w_1\cdots w_{n-1}}$.
We refer to the $i$th state $w_i$ by $\tau(i)$.
Given $i$, $0 \leq i < n$, the prefix $\trace{w_0w_1\cdots w_i}$ is denoted by $\tau_{\leq i}$, the suffix $\trace{w_{i+1},\dots,w_{n-1}}$ is denoted by $\tau_{>i}$, and the length of $\tau$ is denoted by $|\tau|=n$.

\begin{definition}
\label{def:semantics}
A trace $\tau\,{=}\,\trace{w_0w_1\cdots w_{n-1}}$ \emph{satisfies} $\psi{\in} \LL_P$, written 
$\tau \models \psi$, if $\tau,0 \models \psi$ holds, 
where for all $i$, $0\leq i < n$:

\noindent
\begin{tabular}{@{}l@{\:}l@{}}
$\tau,i \models$&$p$ if $p \in \tau(i)$,\\
$\tau,i \models$&$\neg \psi$  if $\tau,i \not\models \psi$, \\
$\tau,i \models$&$ \psi_1 \wedge \psi_2$ if
 $\tau,i \models \psi_1$ and $\tau,i \models \psi_2$, \\
$\tau,i \models$&$ \psi_1 \vee \psi_2$ if
 $\tau,i \models \psi_1$ or $\tau,i \models \psi_2$,\\
$\tau,i \models$&$ \X\psi$ if $i<n{-}1$ and $\tau,i{+}1 \models \psi$,\\
$\tau,i \models$&$ \wX\psi$ if $i = n{-}1$ or $\tau,i{+}1 \models \psi$,\\
$\tau,i \models$&$ \psi_1 \U \psi_2$ if either
 $\tau,i \models \psi_2$, or $i\,{<}\,n{-}1$,\\
 &\qquad
 $\tau,i \models \psi_1$ and
 $\tau,i{+}1\models \psi_1 \U \psi_2$, and\\
$\tau,i \models$&$ \psi_1 \R \psi_2$ if
 $\tau,i \models \psi_2$ and either $\tau,i \models \psi_1$, or\\
 &\qquad
 $i=n{-}1$, or both $i<n{-}1$ and
 $\tau,i{+}1\models \psi_1 \R \psi_2$.
\end{tabular}
\end{definition}
A property $\psi$ is in negation normal form (NNF) if negation appears only in front of propositions.

\paragraph{\LTLf with arithmetic.}
To obtain more expressive constraints, the set of propositional atoms $P$ in \defref{language} can be replaced by atoms from a first-order theory~\cite{GGG22}, and of particular practical interest are arithmetic constraints. In the following, we adopt the logic from \cite{FMPW23} and call it \aLTLf.

To this end, let $V$ be a fixed, finite set of variables such that each $v\,{\in}\,V$ has type $\bools,$ $\ints$ or $\rats$,
with domains $\dom(\bools)=\allowbreak\mathbb B$, $\dom(\ints)=\allowbreak\mathbb Z$ and $\dom(\rats) = \mathbb Q$.
Moreover, let $V' = \{v' \mid v\in V\}$ be a primed copy of $V$; these variables will be used to express \emph{lookahead}, i.e., they are evaluated one instant ahead on the trace, as defined below.
We denote by $\CC(V\cup V')$ the set that consists of all $p\in V$ of sort $\bools$, together with all well-typed linear arithmetic constraints over $V\cup V'$, i.e., all literals of the form $e=e'$, $e\leq e'$, $e < e'$, such that $e$ and $e'$ are well-typed expressions conforming to the grammar 
$e ::= k \mid v \mid v' \mid e + e \mid k\cdot e$
where $v\in V$ and $k \in \mathbb Z \cup \mathbb Q$.
For instance, $x'+1 \geq 2x-y$, and $x = 7y'$ are constraints.

\begin{definition}
The set of arithmetic \LTLf properties $\LL_\AA$
is given by the set of all properties obtained from \defref{language} by
replacing $P$ with $\CC(V\cup V')$.
\end{definition}

For instance, given integer variables $V = \{x,y\}$, $(x' \geq x) \U (x=10)$ and  $\G (2x=y')$ are arithmetic \LTLf properties.
An \emph{assignment} $\alpha$ maps every $v\inn V$ to a value $\alpha(v)$ in the domain of its sort. 
Arithmetic \LTLf properties are evaluated over \emph{arithmetic traces}: such a trace $\tau$ of length $n\,{\geq}\,1$ is a finite sequence
$\trace{\alpha_0, \alpha_1, \dots, \alpha_{n-1}}$ of assignments with domain $V$. 
An expression $e$ is \emph{well-defined} at instant $i$ of $\tau$ if either $i<n{-}1$, or both $i=n{-}1$ and $e$ does not contain $V'$.

\begin{definition}
\label{def:arith:semantics}
If an expression $e$ is well-defined at instant $i$ of a trace $\tau$ of length $n$, 
its \emph{evaluation} $[\tau,i](e)$ at instant $i$ is:\\[1ex] 
\begin{tabular}{@{}r@{\,}l@{\:\:}r@{\,}l@{}}
$[\tau,i](k)$ &$= k$  &
$[\tau,i](e{+}e_2)$ &$=  [\tau,i](e_1){+}[\tau,i](e_2)$\\
$[\tau,i](v^{\prime})$ &$= \tau(i{+}1)(v)$ &
$[\tau,i](k\cdot e)$&$=k\cdot[\tau,i](e)$
\end{tabular}\\[1ex]
Then $\tau \models \psi$ iff $\tau,0 \models \psi$ holds.
For an arithmetic constraint $\psi=e_1 \odot e_2$, we have
$\tau,i \psi$ iff $0 \leq i < n$ and either
$e_1$ or $e_2$ are not well-defined for $\tau$ at $i$, or
$[\tau,i](e_1) \odot [\tau,i](e_2)$ holds. For $p\in V$ of type $\bools$,
$\tau,i \models p$ iff $0 \leq i < n$ and $\alpha(p)=\top$.
All other cases are defined as in \defref{semantics}.
\end{definition}

\noindent
For instance, for integer variables $x$ and $y$ the trace 
\begin{equation}
\label{eq:trace}
\langle \{x\,{\mapsto}\,0, y\,{\mapsto}\,0\}, \{x\,{\mapsto}\,1, y\,{\mapsto}\,0\}, \{x\,{\mapsto}\,2, y\,{\mapsto}\,2\}\rangle
\end{equation}
satisfies both $\G(x'\,{>}\,x) \wedge \F (x\,{=}\,2)$ and  $\G (2x=y')$.
Note that lookahead is evaluated in a weak way, i.e., constraints with lookahead always hold in the last instant.

\paragraph{Monitoring problem.}
Let $\RV = \{\PS, \CS, \CV, \PV\}$ denote the set of distinct \emph{monitoring states} of
permanent satisfaction ($\PS$),
current satisfaction ($\CS$),  
current violation ($\CV$) and permanent violation ($\PV$)~\cite{BaLS10}.

\begin{definition}
A property $\psi$ in $\LL_P$ or $\LL_\AA$ is in
\emph{monitoring state} $s \in RV$ after a trace $\tau$, written $\tau \models \llbracket \psi = s\rrbracket$, if
\begin{compactitem}
\item
$s = \PS$, $\tau \models \psi$, and $\tau\tau' \models \psi$ for every trace $\tau'$;
\item
$s = \CS$, $\tau \models \psi$, but $\tau\tau' \not\models \psi$ for some trace $\tau'$;
\item
$s = \CV$, $\tau \not\models \psi$, but $\tau\tau' \models \psi$ for some trace $\tau'$; 
\item
$s = \PV$, $\tau \not\models \psi$, and $\tau\tau' \not\models \psi$ for every trace $\tau'$.
\end{compactitem}
\end{definition}

For instance, $\tau \models \llbracket \psi = \CS\rrbracket$ means that $\tau$ satisfies  $\psi$ but there is a possible continuation $\tau\tau'$ which does not do so. 
After every trace, a property $\psi$ is in exactly one possible monitoring state. 
Given input $\tau$ and $\psi$,
the (anticipatory) \emph{monitoring problem} 
asks to determine the state $s\in RV$ s.t.  
$\tau \models \llbracket \psi = s\rrbracket$.
E.g., for $\G(x'\,{>}\,x) \wedge \F (x\,{=}\,2)$ and trace \eqref{eq:trace}, $\tau \models \llbracket \psi = \CS\rrbracket$ and $\tau_{\leq 1} \models \llbracket \psi = \CV\rrbracket$, while for 
$\tau' = \trace{\{x\,{\mapsto}\,0\}, \{x\,{\mapsto}\,3\}}$ we obtain  $\tau' \models \llbracket \psi = \PV\rrbracket$.
For arithmetic \LTLf, the monitoring problem is known to be undecidable~\cite[Thm. 6]{FMPW23}.

\paragraph{Automata-based monitoring.}
For a propositional \LTLf property $\psi \in \LL_P$, a common monitoring approach~\cite{DDMM22} is to construct a deterministic finite automaton (DFA) $\DFA$  for $\psi$, where each state of $\DFA$ is associated with a monitoring value $rv(q)$ as follows:
if $q$ is final then $rv(q)\,{=}\,\PS$ if all states reachable from $q$ are final, and $rv(q) = \CS$ otherwise; and 
if $q$ is non-final then $rv(q) = \PV$ if no state reachable from $q$ is final, and $rv(q) = \CV$ otherwise.
Every trace $\tau$ leads to a unique DFA state $q$, and 
the monitoring state of $\tau$ wrt. $\psi$ is $rv(q)$, that is, $\tau \models \llbracket \psi = rv(q)\rrbracket$.

\begin{example}
\label{exa:DFA}
For $\psi = \wX \G a \wedge \F b$, a DFA is as follows:\\
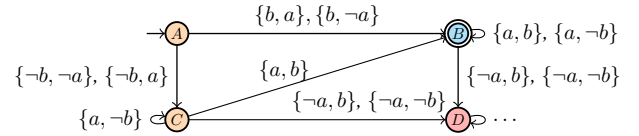
\begin{tikzpicture}[node distance=62mm]
\tikzstyle{state}=[draw, circle, inner sep=1.5pt, line width=.7pt, scale=.6]
\tikzstyle{edge}=[draw, ->, line width=.5pt]
\tikzstyle{caption}=[scale=.9]
\tikzstyle{node} = [draw,rectangle split, rectangle split parts=2,rectangle split horizontal, rectangle split draw splits=true, inner sep=3pt, scale=.65, rounded corners=2pt]
\tikzstyle{goto} = [->]
\tikzstyle{action}=[scale=.8, black]
\tikzstyle{final}=[double]
\tikzstyle{pscolor}=[fill=green!30]
\tikzstyle{cscolor}=[fill=cyan!30]
\tikzstyle{pvcolor}=[fill=red!30]
\tikzstyle{cvcolor}=[fill=orange!30]
 \node[state, cvcolor] (A) {$A$};
 \node[state, right of =A, final, cscolor] (B) {$B$};
 \node[state, below of=A, yshift=43mm, cvcolor] (C) {$C$};
 \node[state, right of=C, pvcolor] (E) {$D$};
\draw[edge] ($(A) + (-.4,0)$) -- (A);
\draw[edge] (A) -- node[action, above] {$\{b, a\}, \{b, \neg a\}$} (B);
\draw[edge] (A) -- 
node[action, left] {$\{\neg b, \neg a\}$, $\{\neg b, a\}$} 
(C);
\draw[->] (B) to[loop right, looseness=8]
  node[action, right, anchor=west] {$\{a, b\}$, $\{a, \neg b\}$} (B);
\draw[edge] (B) -- 
  node[action, right, anchor=west] {$\{\neg a, b\}$, $\{\neg a, \neg b\}$} (E);
\draw[->] (C) to[loop left, looseness=8] node[action, left, anchor=east] {$\{a, \neg b\}$} (C);
\draw[->] (C) to node[action, left, anchor=east, pos=.5, yshift=1mm] {$\{a, b\}$} (B);
\draw[->] (C) to 
  node[action, above,near end, xshift=-2mm] {$\{\neg a, b\}$, $\{\neg a, \neg b\}$} (E);
\draw[->] (E) to[loop right, looseness=8] node[action, right] {$\dots$}  (E);
\end{tikzpicture}
Here the colors of states indicate monitoring states: $A$ and $C$ correspond to $\CV$, $B$ to $\CS$ and $D$ to $\PV$. 
E.g. the trace $\{a \},$ $\{a \}, \{a, b\}$ leads to state $B$ and has monitoring state $\CS$.
\end{example}

However, for arithmetic \LTLf this approach was shown to be incorrect,  essentially, as due to lookback not all paths in the DFA are feasible with all variable valuations.
To see that, we first note that for technical reasons, the automata-based approach for \aLTLf~\cite{FMPW23} uses look\emph{back} instead of lookahead:
it introduces variables $\pre V = \{\pre{v} \mid v\in V\}$ that relate to the \emph{previous} trace instant; besides that, constraints are defined in the same way.
Lookback and lookahead are equivalent as far as expressivity is concerned: e.g. $\G(x'\,{>}\,x) \wedge \F (x\,{=}\,2)$ is equivalent to $\psi=\wX\G(x\,{>}\,\pre{x}) \wedge \F (x\,{=}\,2)$.
The DFA for $\psi$ corresponds to that in \exaref{DFA}, with $a$ replaced by $x\,{>}\,\pre{x}$ and $b$ by $x=2$.
Now, the trace $\trace{\{x\,{\mapsto}\,0\}, \{x\,{\mapsto}\,3\}}$ leads to state $C$, though its monitoring state is $\PV$ and not $\CV$.
To alleviate this problem, for \aLTLf the automata-based approach was extended with \emph{constraint graphs} (CGs), a data structure that is obtained by unrolling $\DFA[\psi]$ while computing the effect that constraints have on variables. They allow to compute first-order formulas expressing whether a (non)final DFA state is reachable.
In the monitoring algorithm, both the DFA and the CGs are constructed in a preprocessing phase, after which at each trace instant the DFA state is updated and the reachability conditions from the CG are evaluated. 
However, CGs are in general infinite so that the method need not terminate.
For details we refer to~\cite{FMW22a}.

\section{Progression Monitoring}
\label{sec:progression}

This section presents an alternative, \emph{lazy} monitoring approach based on progression~\cite{BacchusK00} and satisfiability checking for propositional \LTLf.

\paragraph{Propositional \LTLf.}
We first define  a finite-trace version of progression, which differs from the original \cite[Table 2]{BacchusK00} in that we assume a dedicated proposition $\last$ that serves to mark the final state of a trace. A trace $\trace{w_0,w_1,\dots, w_{n-1}} \in (2^{P \cup \{\last\}})^*$ is called \emph{well-formed} if $\last \not \in w_i$ for $0\leq i < n{-}1$. 

\begin{definition}[Progression]
\label{def:progression}
Let $\psi \in \LL_P$ be in NNF and $w\subseteq P \cup \{\last\}$.
The progression $\psi^+(w)$ of $\psi$ with respect to $w$ is defined as follows, where $p$ is a proposition in $P$:
\begin{compactitem}
\item 
$\top^+(w) = \top$ and $\bot^+(w) = \bot$,
\item
$p^+(w) = \top$ if $p\in w$ and $p^+(w) = \bot$ otherwise;
\item 
$(\neg p)^+(w) = \bot$ if $p\in w$ and $(\neg p)^+(w) = \top$ otherwise;
\item 
if $\psi = \psi_1 \wedge \psi_2$ then $\psi^+(w) = \psi_1^+(w) \wedge \psi_2^+(w)$;
\item 
if $\psi = \psi_1 \vee \psi_2$ then $\psi^+(w) = \psi_1^+(w) \vee \psi_2^+(w)$;
\item 
if $\psi = \X \psi_1$ then $\psi^+(w) = \psi_1$ if $\last \not\in w$, and $\psi^+(w) = \bot$ otherwise;
\item 
if $\psi = \wX \psi_1$ then $\psi^+(w) = \psi_1$ if $\last \not\in w$, and $\psi^+(w) = \top$ otherwise;
\item 
if $\psi = \psi_1 \U \psi_2$ then $\psi^+(w) = \psi_2^+(w) \vee (\psi_1^+(w) \wedge \psi)$ if $\last \not\in w$, and $\psi^+(w) = \psi_2^+(w)$ otherwise.
\item 
if $\psi = \psi_1 \R \psi_2$ then $\psi^+(w) = \psi_1^+(w) \wedge (\psi_2^+(w) \vee \psi)$ if $\last \not\in w$, and $\psi^+(w) = \psi_1^+(w)$ otherwise.
\end{compactitem}
\end{definition}

For instance, for $\psi = \wX \G a \wedge \F b$ we have $\psi^+(\{a\})=\G a \wedge \F b$ but $\psi^+(\{b\})=\G a$.
Intuitively, the progression of a property $\psi$ with respect to $w$
amounts to the ``remainder'' of $\psi$ that still has to be checked 
after observing $w$ in a trace, i.e., the remaining obligation after one monitoring step.
Next, we present an RV-LTL monitoring algorithm that exploits this correspondence.
We briefly outline its functioning:
the procedure takes as input a property $\psi \in \LL_P$ and a trace $\tau$ over $P$.
Initially, $\psi$ is transformed into negation normal form $\phi$ before a monitoring loop iterates over the trace.
In each iteration, $\phi$ is progressed twice, assuming that the trace ends (Line \ref{line:progression:progresslast}) resp. that it does not (Line \ref{line:progression:progress}).
The latter is equivalent\footnote{Here $\equiv$ refers to equivalence of propositional formulas without variables, i.e. using $\top \wedge \bot \equiv \bot$, $\top \vee \bot \equiv \top$ etc.} to $\top$ iff $\tau_{\leq i} \models \psi$, as we will show below.
A respective case distinction is performed in Line \ref{line:progression:if}.
In the \emph{else} case $\tau_{\leq i} \not\models \psi$ holds, so the property is at least currently violated. The algorithm checks whether the remaining monitoring obligation $\phi$ is satisfiable. If so, the monitoring state $s$ is  $\CV$, otherwise $s=\PV$ (Line \ref{line:progression:vio}). The \emph{if} case is dual (Line \ref{line:progression:sat}), checking satisfiability of $\neg \phi$ to test whether a future continuation can violate $\phi$.

\newcommand{\chicurr}{\chi}
\begin{algorithm}
\caption{Procedure \textsc{ProgMonitor}$(\psi,\tau)$}
\label{alg:progression:monitoring}
\begin{algorithmic}[1]
  \State $\phi \gets \mathit{NNF}(\psi)$
  \While{$0\leq i < |\tau|$} \label{line:progression:loop}
    \State $\chi \gets \phi^+(\tau(i)\cup\{\last\})$ \label{line:progression:progresslast}
    \State $\phi \gets \phi^+(\tau(i))$ \label{line:progression:progress}
    \If{$\chi \equiv \top$} \label{line:progression:if}
        \State $s \gets \CS$ \textbf{if} $\neg\phi$ is satisfiable \textbf{else} $\PS$ \label{line:progression:sat}
    \Else
        \State $s \gets \CV$ \textbf{if} $\phi$ is satisfiable \textbf{else} $\PV$ \label{line:progression:vio}
    \EndIf
    \State \textbf{output} $s$
    \EndWhile
\end{algorithmic}
\end{algorithm}

\begin{example}
Consider again the property $\psi = \wX \G a \wedge \F b$ which is already in NNF,
and the trace $\tau =\trace{\{a\}, \{a\}, \{a, b\}}$ from \exaref{DFA}.
Let $w = \{a\}$ and 
$w_l = \{a, \last\}$.
\begin{compactenum}
\item 
We have $\psi^+w_l = (\wX \G a)^+w_l \wedge (\F b)^+w_l = \top \wedge \bot \equiv \bot$ and $\psi^+w = (\wX \G a)^+w \wedge (\F b)^+w = \G a \wedge \F b$. As $\phi := \G a \wedge \F b$ is satisfiable, the monitoring state is $\CV$.
\item In the next iteration, $\phi^+w_l = (\G a)^+ w_l\wedge (\F b)^+w_l = \top \wedge \bot \equiv \bot$ while $\phi^+w= (\G a)^+w\wedge (\F b)^+w = \phi$, so the monitoring state is again $\CV$.
\item Finally, $\phi^+\{a, b, \last\}$ evaluates to $\top$ and $\phi^+\{a, b\} = (\G a)^+\{a, b\} \wedge (\F b)^+\{a, b\} = (\G a) \wedge \top \equiv \G a$. As $\neg\phi = \F \neg a$ is satisfiable, the monitoring state is $\CS$.
\end{compactenum}
\end{example}

The next paragraphs show correctness of Alg.~\ref{alg:progression:monitoring}.
To that end, progression is extended to traces as usual:
Given a trace $\tau =\trace{w_0w_1 \dots w_n}$, we write $\psi^+(\tau)$ for the repeated progression $\psi^+(w_0)^+(w_1)^+ \cdots (w_n)$.
The correctness proof below uses the following key property of progression (cf.~\cite[Theorem 4.3]{BacchusK00}).

\begin{restatable}{lemma}{lemmaprog}
\label{lem:progression}
Let $\psi \in \LL_P \cup \{\bot,\top\}$,  $\tau  \in (2^{P \cup \{\last\}})^*$ be a well-formed trace, and $i\geq 0$. 
Then $\tau \models \psi$ iff either $i \geq |\tau|-1$ and $\psi^+(\tau)\equiv \top$, or 
$\tau_{>i} \models \psi^+(\tau_{\leq i})$.
\end{restatable}

\begin{restatable}[Correctness]{theorem}{correctness}
\label{thm:correctness}
Given inputs $\psi \in \LL_P$ and a trace $\tau \in (2^P)^*$, Alg.~\ref{alg:progression:monitoring}
outputs the monitoring state with respect to $\psi$ for every position of $\tau$, i.e., the sequence of all $s_i\in \RV$ such that $\psi \models \llbracket \tau_{\leq i} = s_i \rrbracket$ for all $0 \leq i < |\tau|$.
\end{restatable}

It can be observed that for safety and cosafety properties, Alg.~\ref{alg:progression:monitoring} simplifies slightly, as follows.

\begin{remark}
    \label{rem:safety}
\LTLf safety properties $\psi$ are characterized by the fact that every violating trace $\tau$ has a \emph{bad prefix}, i.e. there is some $i$, $0\,{\leq}\,i\,{<}\,|\tau|$, such that $\tau_{\leq i}\tau'$ violates $\psi$ for every trace $\tau'$~\cite{KupfermanV01}. Now let $\psi$ be such a property. Then for every $\tau$ such that $\tau \not\models \psi$, by definition $\tau$ \emph{permanently} violates $\psi$. Hence the monitoring state $\CV$ never occurs. It follows that the satisfiability check in Line~\ref{line:progression:vio} of Alg.~\ref{alg:progression:monitoring} can be omitted if a property $\psi$ is known to be in the safety fragment, e.g. because it satisfies syntactic criteria such as being in NNF and containing only the temporal operators $\R$ and $\wX$~\cite{Geatti2025SafetyAL}.
Dually, cosafety properties for which all satisfying traces have a \emph{good prefix} never assume the monitoring state $\CS$, so the satisfiability check in Line~\ref{line:progression:sat} can be skipped. 
\end{remark}

We briefly comment on computational complexity. Given a formula $\psi$,  a single progression step is linear in the size of $\psi$ but can double the size of the formula. Hence, for a trace $\tau$, the size of the progressed formula $\psi^+(\tau)$  is in $O(|\psi|\cdot 2^{|\tau|})$ \citep[Section 7.1]{BauerECAI2010}.
Since \LTLf satisfiability checking is in PSPACE, every satisfiability check in Alg.~\ref{alg:progression:monitoring} can require space  polynomial in $|\psi|$ and exponential in $|\tau|$. 
For the (infinite) set of all finite traces on $P$, the set of all corresponding progressed formulas from $\psi$, namely
$\{\psi^+(\tau')\mid \tau' \in (2^P)^*\}$,
is finite if suitable simplifications are applied%
, and has size in $\smash{O(2^{2^{|\psi|}})}$, as each of its elements corresponds to a state in the DFA construction~\cite[Thms. 2 and 3]{GiacomoFLVX022}. This means that, when monitoring an evolving trace, progressed formulas will eventually repeat. Hence, an implementation of Alg.~\ref{alg:progression:monitoring} can improve performance by caching 
the results of 
satisfiability checks.

\section{Arithmetic Progression Monitoring}
\label{sec:arithmetic:progression}

Next, we extend  progression monitoring to \LTLf with arithmetic.
To define a modified notion of progression for $\LL_\AA$ that uses information about trace length, we assume now a fresh boolean variable $\last\not\in V$, and call
a trace $\trace{\alpha_0 \dots \alpha_n}$ over $V\cup\{\last\}$ \emph{well-formed} if $\alpha_n(\last)=\top$ and  $\alpha_i(\last)=\bot$ for all $i$ with $0\leq i < n$.

\begin{definition}[Progression of constraints]
For a constraint $c \in \CC(V \cup V')$ and an assignment $\alpha$ with domain $V$,
let $\mathit{progress}(c, \alpha)$ denote the constraint $c\alpha\gamma$ for the substitution $\gamma$ with domain $V'$ that sets $\gamma(v') = v$.
\end{definition}

For instance, for $c = (x'>x)$ and $\alpha = \{x \mapsto 3\}$, we have $\mathit{progress}(c, \alpha) = x>3$.
The progression of a constraint over variables $V \cup V'$ is thus again a constraint $c'$, in which variables in $V$ are replaced by their value according to $\alpha$, and variables in $V'$ are replaced by their counterpart in $V$, so $c'$ is a constraint over variables $V$.

\begin{definition}[Arithmetic progression]
\label{def:progression:arith}
Let $\psi \in \LL_\AA$ be in NNF
and $\alpha$ an assignment with domain $V \cup \{\last\}$.
The progression $\psi^+(\alpha)$ of $\psi$ with respect to $\alpha$ is defined as follows:
\begin{compactitem}
\item
if $\psi \in \CC(V \cup V')$ then $\psi^+(\alpha) = \top$ if $\alpha(\last) = \top$ and $\psi$ contains $V'$, otherwise $\psi^+(\alpha) = \mathit{progress}(\psi, \alpha)$;
\item 
if $\psi = \psi_1 \vee \psi_2$ then $\psi^+(\alpha) = \psi_1^+(\alpha) \vee \psi_2^+(\alpha)$;
\item 
if $\psi = \psi_1 \wedge \psi_2$ then $\psi^+(\alpha) = \psi_1^+(\alpha) \wedge \psi_2^+(\alpha)$;
\item 
$(\X \psi_1)^+(\alpha)$ is $\psi_1$ if $\alpha(\last) = \bot$, and $\bot$ otherwise;
\item 
$(\wX \psi_1)^+(\alpha)$ is $\psi_1$ if $\alpha(\last) = \bot$, and $\top$ otherwise;
\item 
if $\psi = \psi_1 \U \psi_2$ then $\psi^+(\alpha) = \psi_2^+(\alpha) \vee (\psi_1^+(\alpha) \wedge \psi)$ if $\alpha(\last)=\bot$, and $\psi^+(\alpha) = \psi_2^+(\alpha)$ otherwise;
\item 
if $\psi = \psi_1 \R \psi_2$ then $\psi^+(\alpha) = \psi_1^+(\alpha) \wedge (\psi_2^+(\alpha) \vee \psi)$ if $\alpha(\last)=\bot$, and $\psi^+(\alpha) = \psi_1^+(\alpha)$ otherwise.
\end{compactitem}
\end{definition}

As in the propositional case, progression can be naturally extended to traces:
Given a trace $\tau =\trace{\alpha_0\alpha_1 \dots \alpha_{n-1}}$, let $\psi^+(\tau) = \psi^+(\alpha_0)^+(\alpha_1)^+ \cdots (\alpha_{n-1})$.

\begin{example}
\label{exa:arith:progression}
Consider $\psi= \G (x' > x) \wedge \F (x=2)$ and let 
$\alpha_0 = \{ x\,{\mapsto}\,0,\last\,{\mapsto}\,\bot\}$ and 
$\alpha_1 = \{ x\,{\mapsto}\,3,\last\,{\mapsto}\,\bot\}$.
We have 
$\psi^+(\alpha_0) = (x > 0) \wedge \G (x' > x) \wedge \F (x=2)$;  this formula is satisfiable. On the other hand,
$\psi^+(\alpha_1) = (x > 3) \wedge \G (x' > x) \wedge \F (x=2)$ is unsatisfiable.
For the respective assignments that indicate that the trace ends, 
$\psi^+(\{ x\,{\mapsto}\,0,\last\,{\mapsto}\,\top\}) = \psi^+(\{ x\,{\mapsto}\,3,\last\,{\mapsto}\,\top\}) = \bot$.
\end{example}

Progression monitoring for \aLTLf can now be defined similarly to the propositional case (cf. Alg.~\ref{alg:progression:monitoring}), with the only differences that the modified notion of progression is used (\defref{progression:arith}), and $\tau(i) \cup\{\last\}$ in Line~\ref{line:progression:progresslast} resp. $\tau(i)$ in Line~\ref{line:progression:progress} must be replaced by $\tau(i) \cup\{\last \mapsto \top\}$ resp. $\tau(i) \cup\{\last \mapsto \bot\}$).

\begin{example}
For $\psi$ from \exaref{arith:progression} and the trace
$\tau = \trace{\{ x \mapsto 0,\last \mapsto \bot\}\{ x \mapsto 3,\last \mapsto \bot\}}$,
from the observations in \exaref{arith:progression} it  follows that
Alg.~\ref{alg:progression:monitoring} outputs $\CV$ for instant $0$ and $\PV$ for instant 1. Indeed,
$\tau_{\leq 0} \models [\psi = \CV]$ but 
$\tau_{\leq 1} \models [\psi = \PV]$.
\end{example}

In order to prove correctness, we need the following property, similar as for the propositional case:

\begin{restatable}{lemma}{lemmaprogressionarith}
\label{lem:progression:arith}
Let $\psi \in \LL_\AA \cup \{\bot,\top\}$,  $\tau$ be a well-formed arithmetic trace over $V$, and $i\geq 0$. 
Then $\tau \models \psi$ iff either $i \geq |\tau|$ and $\psi^+(\tau)\equiv \top$, or 
$\tau_{>i} \models \psi^+(\tau_{\leq i})$.
\end{restatable}

Correctness of progression monitoring for \aLTLf is stated next, it can be proven similarly as \thmref{correctness} but using \lemref{progression:arith}.

\begin{restatable}[Correctness]{theorem}{theoremprogressionarith}
\label{thm:correctness:progressive:arith}
On inputs $\psi \in \LL_\AA$ and an arithmetic trace $\tau$,
if Alg.~\ref{alg:progression:monitoring} terminates then it
outputs the monitoring state with respect to $\psi$ for every position of $\tau$.
\end{restatable}

\paragraph{Decidability.}
Satisfiability of arithmetic \LTLf is undecidable, so that \progmon need not terminate on $\LL_\AA$ in general. However, several relevant, satisfiable fragments of  \aLTLf are known~\cite{GeattiGGW23}.
These results can be used to identify fragments of arithmetic \LTLf for which progression monitoring is guaranteed to terminate, and hence monitoring is \emph{solvable} in the sense that it always produces a definitive verdict.
To this end, we recall special classes of constraints defined in the literature~\cite{DemriD07,GeattiGGW23}:
A \emph{monotonicity constraint} (MC) over rational variables $V$ has the form $p \odot q$ where $p,q\in {\mathbb Q\,{\cup}\,V}$
and $\odot$ is one of $=, \neq, \leq$, or $<$. An \emph{integer periodicity constraint} (IPC) over integer variables $V$ has the form 
 $x = y$, $x \odot d$ for $\odot \in \{=,\neq, <, >\}$, $x \equiv_k y + d$, or $x \equiv_k d$, for variables $x,y$ with domain $\mathbb Z$ and $k,d\in \mathbb N$.
 MC properties resp. IPC properties are \aLTLf properties where all constraints are MCs resp. IPCs. 
 \citeauthor{FMPW23} (\citeyear{FMPW23}) showed that automata-based monitoring is solvable for \aLTLf properties whose constraints are all MCs, or all IPCs. E.g.,$\G(t \geq t' \wedge t \leq 30)$ is an MC property.
One can prove that in these cases also progression monitoring is guaranteed to terminate:

 \begin{restatable}{theorem}{decidabilityMCIPC}
 \label{thm:decidability:progressive}
 Progression monitoring is guaranteed to terminate for MC and IPC properties.
 \end{restatable}

\citeauthor{GeattiGGW23} (\citeyear{GeattiGGW23}) do in fact  prove decidability of satisfiability checking for all \aLTLf formulas whose \emph{iteration conditions} are either all MCs or all IPCs; these are a superset of MC resp. IPC properties. Here the set of iteration conditions of an \LTLf formula $\varphi$ in NNF consists of all constraints that occur in $\varphi_1$ for any subformula $\varphi_1 \U \varphi_2$ of $\varphi$, or in $\varphi_2$ for any subformula $\varphi_1 \R \varphi_2$ of $\varphi$.
 Dually, one can define the set of \emph{co-iteration} conditions of $\varphi$ as the set of constraints that occur in $\varphi_2$ for a subformula $\varphi_1 \U \varphi_2$ of $\varphi$, or in $\varphi_1$ for a subformula $\varphi_1 \R \varphi_2$ of $\varphi$.
 We can exploit these results to obtain new insights into solvability conditions for monitoring, but given the structure of the progression monitoring algorithm, we need to make sure that not only the properties $\phi_i$ in Alg.~\ref{alg:progression:monitoring}, but also their negations are in a decidable class. To this end, we exploit our observations about safety (Rem.~\ref{rem:safety}).

 \begin{restatable}{theorem}{decidabilitysafety}
 \label{thm:decidability:progressive:safety}
Progression monitoring terminates for every \aLTLf property $\psi$ in NNF that satisfies one of the following:
\begin{inparaenum}
\item $\psi$ is a co-safety property and its iteration conditions are either all MCs, or all IPCs;
\item $\psi$ is a safety property and its co-iteration conditions are all MCs, or all IPCs.
\end{inparaenum}
 \end{restatable}

By slightly adapting the results from \citeauthor{Geatti2025SafetyAL} \citeyear{Geatti2025SafetyAL}, one can e.g. show that all \aLTLf formulas in NNF that use only the temporal operators $\wX$ and $\R$, and where constraints with lookahead are not negated, are safety formulas (cf. the appendix). 
 Thus \thmref{decidability:progressive:safety} covers e.g. 
  the safety property $(x \equiv_3 0)\R (x+y=x')$ where the co-iteration condition is IPC, but also the co-safety property $(x<y) \U (\X(x+y=27))$.

\section{Evaluation}
\label{sec:experiments}

We implemented progression monitoring for both propositional and arithmetic \LTLf in C++, using BLACK~\cite{GeattiGMV24} for satisfiability checking.
The results of satisfiability checks are cached for later use.
For comparison, we implemented also propositional \dfamon using either Lydia~\cite{Lydia}, or Spot with MT\-DFAs~\cite{DuretLutzZPGV25} as backend to construct automata; as well as \dfamon with CGs
for the arithmetic case, using in addition Z3~\cite{Z3} for constraint handling.
For the propositional case, we  implemented in addition a combined approach where the DFA construction is done with a separate thread, while at the same time the progression approach is used to start monitoring the trace. Once the DFA construction is completed, monitoring is continued with the automaton.
All experiments were run on a system equipped with a 10-core 13th Gen Intel Core i7-1365U processor and 32 GB of main memory.

\paragraph{Propositional \LTLf.}
We consider two complementary benchmark sets: formulas from the synthesis competition as rather syntactic problems (SC), and declarative specifications of business processes, an area where monitoring is very relevant in practice~\cite{CCGD26} (PM).
For (SC), we took all formulas of the synthesis competition 2025~\cite{syntcomp} that could be transformed to \LTLf (1744 formulas) and generated for each formula three traces of length 500: a random one, one by walking along the DFA (if it was computable by Lydia), and one by subsequently applying progression to the formula and choosing an assignment that satisfies the current requirements. In this way, 4471 formula/trace pairs were obtained.
(PM) was constructed by taking all event logs from the website of the IEEE Task Force on Process Mining,\footnote{\url{https://www.tf-pm.org/resources/logs}} and used Declare4Py~\cite{DonadelloRMS22} to automatically mine \Declare constraints from the event log. The conjunction of these constraints, transformed into \LTLf, was monitored against 100 random traces from the respective event log, resulting in 3637 formula/trace pairs.

We ran all monitoring algorithms with a timeout of 180 seconds for each formula/trace pair, determine the monitoring state for each trace prefix. 
All techniques produced the same monitoring state.
Tab.~\ref{tab:propositional} reports the results of \dfamon based on automata generated by Lydia or Spot, \progmon, and the combined approach (with Spot) on both benchmark sets. The median times show that \dfamon with Spot is typically faster than \dfamon with Lydia, and both are faster than progression. The average time is often higher for those monitors that experience fewer timeouts, solving more hard problems.
\commenttkbox[noinline]{Description of table moved to the table caption. Can we say anything more about it there or here?}
\begin{table}[t]
\begin{footnotesize}
\begin{tabular}{@{}l@{\:}|@{\:}c@{\:}c@{\:}c@{\:}|@{\:}c@{\:}c@{\:}c@{\:}|@{\:}c@{\:}c@{\:}c@{\:}|@{\:}c@{\:}c@{\:}c@{}}
&
\multicolumn{3}{@{\:}c@{\:}|@{\:}}{\textbf{Lydia}} &
\multicolumn{3}{@{\:}c@{\:}|@{\:}}{\textbf{Spot}} &
\multicolumn{3}{@{\:}c@{\:}|@{\:}}{\textbf{progression}} &
\multicolumn{3}{@{\:}c}{\textbf{combined}} \\
&
suc & avg & med &
suc & avg & med &
suc & avg & med &
suc & avg & med \\
\hline
SC &
3506 & 1.18 & 0.05 &
3550 & 2.11 & 0.01 &
3710 & 2.96 & 0.07 &
4432 & 2.03 & 0.12 \\
\hline
PM & 
2621 & 0.80 & 0.09 &
2621 & 2.11 & 0.01 &
3535 & 2.71 & 0.19 &
3577 & 4.42 & 0.11
\end{tabular}
\end{footnotesize}
\caption{Experiments for propositional \LTLf \added{on the SC and PM benchmarks (with 4471 and 3637 instances, resp.)}.
\added{
For each technique, the number of successes, i.e., non-timeouts, is reported (suc), as well as the average (avg) and median (med) monitoring time for successful runs in seconds.
}
}
\label{tab:propositional}
\end{table}
\begin{figure}
\resizebox{.48\textwidth}{!}{
\begin{tikzpicture}
\begin{axis}[
	width=.55\textwidth,
	height=.3\textwidth,
	xlabel={time (seconds)},
	ylabel={number of instances},
	legend style={at={(0.9,0.7)},anchor=east, font=\footnotesize},
    enlarge x limits = 0.0,
    enlarge y limits = 0.05,
]
\addplot+[orange, no marks] table {DFA_cactus_SC.dat};
\addlegendentry{Lydia}

\addplot+[red, no marks] table {spot_cactus_SC.dat};
\addlegendentry{Spot}

\addplot+[blue, no marks] table {prog_cactus_SC.dat};
\addlegendentry{progression}

\addplot+[darkgreen, no marks] table {spot_combined_cactus_SC.dat};
\addlegendentry{combined}
\end{axis}
\end{tikzpicture}
}
\\
\resizebox{.48\textwidth}{!}{
\begin{tikzpicture}
\begin{axis}[
	width=.55\textwidth,
	height=.3\textwidth,
	xlabel={time (seconds)},
	ylabel={number of instances},
	legend style={at={(0.9,0.62)},anchor=east, font=\footnotesize},
    enlarge x limits = 0.0,
    enlarge y limits = 0.05,
]
\addplot+[orange, no marks] table {DFA_cactus_PM.dat};
\addlegendentry{Lydia}

\addplot+[red, no marks] table {spot_cactus_PM.dat};
\addlegendentry{Spot}

\addplot+[blue, no marks] table {prog_cactus_PM.dat};
\addlegendentry{progression}

\addplot+[darkgreen, no marks] table {spot_combined_cactus_PM.dat};
\addlegendentry{combined}

\end{axis}
\end{tikzpicture}
}
\caption{Number of instances solved within time limit for the SC (above) and PM benchmarks (below).}
\label{fig:cactus}
\end{figure}
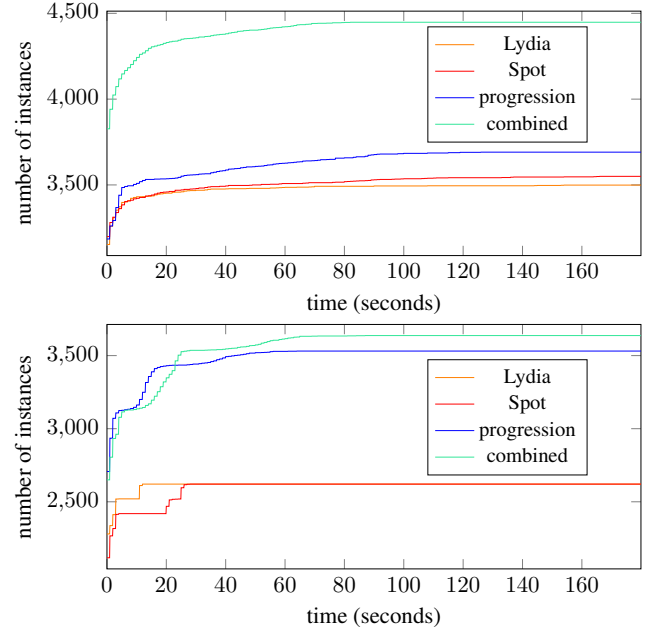
\figref{cactus} shows cactus plots or both benchmark sets, counting how many instances were solved within a time budget. It illustrates that \progmon solves somewhat more benchmarks that \dfamon, but the combined approach (using the DFA/Spot monitor and \progmon in parallel) handles even more problems. \figref{scatter} compares the monitoring time \dfamon with Spot and \progmon on those SC problems where both succeed: the latter typically requires more time, on average twice as much (a respective comparison for (PM) is similar). 
However, for large formulas, \dfamon often fails to construct the automaton. \figref{nim} shows the monitoring time for the formulas describing the two-player game Nim from the (SC) set, where the times required by \dfamon increase fast and the tool times out for $n > 8$.
Timeouts were counted if the technique fails to process the entire trace within 180 seconds. For DFA monitoring timeouts almost always occur during DFA construction (in general, more than 90\% of the total time is spent on automata construction). 
Moreover, interestingly in all cases where \dfamon succeeded but \progmon timed out, the final monitoring state was either $\PV$ or $\PS$. This is likely due to the fact that in these cases BLACK must prove \emph{un}satisfiability, which is more difficult than proving satisfiability with the implemented tableaux method.
In progression monitoring, on average 62\% of the time was spent on satisfiability checks.
For the combined approach, in 60\% (SC) resp. 51\% (PM) of all benchmarks, the phase of DFA monitoring was reached, i.e., the remaining cases were solved by \progmon alone.

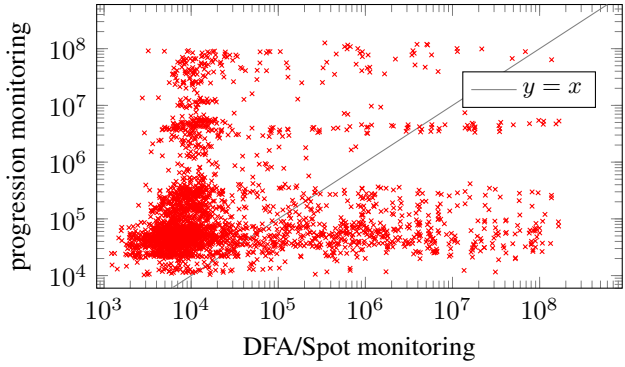
\begin{figure}
\begin{tikzpicture}
\begin{axis}[
	width=.48\textwidth,
	height=.3\textwidth,
	xmode=log,
	ymode=log,
	xlabel={DFA/Spot monitoring},
	ylabel={progression monitoring},
	legend pos=north west,
    enlarge x limits = 0.03,
    enlarge y limits = 0.0,
	legend style={at={(0.95,0.7)},anchor=east, font=\footnotesize}
]

\addplot+[gray,  domain=6e3:6e8, no markers] {x};
\addlegendentry{$y=x$}

\addplot+[red, only marks, mark=x, mark size=1.2pt] table {dfa_vs_prog_SC_all.dat};




\end{axis}
\end{tikzpicture}
\caption{Comparison of the computation time in $\mu$s required by \dfamon and \progmon on SC benchmarks.}
\label{fig:scatter}
\end{figure}

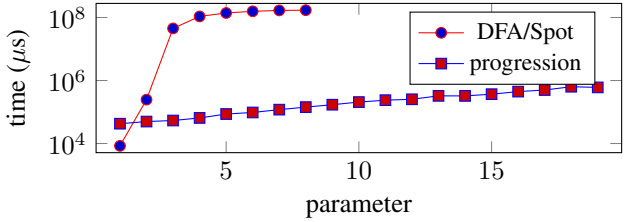
\begin{figure}
\begin{tikzpicture}
\begin{axis}[
	width=.48\textwidth,
	height=.2\textwidth,
	xmode=linear,
	ymode=log,
	xlabel={parameter},
	ylabel={time ($\mu$s)},
	legend pos=north west,
    enlarge x limits = 0.05,
    enlarge y limits = 0.05,
	legend style={at={(0.95,0.7)},anchor=east, font=\footnotesize}
]

\addplot+[red, mark size=2pt] table {nim_time_spot.dat};
\addlegendentry{DFA/Spot}
\addplot+[blue,  mark size=2pt] table {nim_time_prog.dat};
\addlegendentry{progression}
\end{axis}
\end{tikzpicture}
\caption{Monitoring parametric formulas describing the two-player game Nim.}
\label{fig:nim}
\end{figure}

\paragraph{Arithmetic \LTLf.} 
We are not aware of any comprehensive set of \aLTLf benchmarks, and hence
collected formulas from the literature: (A) all examples from the tool \emph{ada}~\cite{FMPW23}, (B) example instances from all formula families of the BLACK repository,\footnote{\url{https://github.com/black-sat/black}} and (C) examples from the Lola repository, a total of 36 formula/trace pairs.
We compare arithmetic progression monitoring with the tool \textit{ada}~\cite{FMW22a}. However, since the latter is written in Python and uses a naive automata constuction, we re-implemented the approach in C++ using Lydia DFAs.  All benchmarks were run with a timeout of 180 seconds.
A detailed table can be found in the appendix, Tab.~\ref{tab:arithmetic}.
In summary, \progmon succeeds to monitor 32, \dfamon 19, and \textit{ada} 10 problems, and there are no problems where \dfamon or \textit{ada} succeeds but \progmon does not. On examples where both succeed, \progmon requires on average 28\% of the time needed by \dfamon (resp. 12\% of the time needed by \textit{ada}).
On many examples, the automata-based approach times out during DFA/CG construction, whereas progression monitoring succeeds. Moreover, for all benchmarks solved by both approaches, progression monitoring is faster. For the automata-based approach, the DFA/CG construction requires more than 99\% of the overall time.

\section{Discussion and Conclusions}
\label{sec:discussion}

To conclude, we compare automata-based to progression-based monitoring from multiple viewpoints.

For monitoring of propositional \LTLf formulas, our evaluation suggests that if the DFA can be computed, \dfamon is more efficient. However, for many, in particular large, formulas as they appear, e.g., in monitoring of business processes (cf. PM benchmark), the DFA computation did not terminate with either automata construction library, so that \progmon is clearly preferrable. 
If a specification is monitored repeatedly against multiple traces, it may be worth investing more time in the monitor construction, though.
In any case, our combined monitoring approach can serve as a useful tradeoff, as it parallelizes the DFA computation and the progression approach, jumping to the faster automata-based approach as soon as possible. The evaluation indeed confirms its ability to solve considerably more benchmarks than either approach alone, and to be about 20\% faster than \progmon on average.
\sarahtodo{find percentage}

For the arithmetic case, the advantage of the progression-based approach is much more prominent: \dfamon in fact often fails to terminate in the preprocessing phase, while the lazy and trace-tailored approach of \progmon yields fewer timeouts and is faster in all benchmarks.

In terms of computational complexity, in the propositional case, the DFA construction needed by automata-based monitoring is doubly exponential in the size of the formula, while a subsequent monitoring step is comparatively cheap if the constructed DFA can be kept in memory and smart data structures such as MTDFAs~\cite{Mona,DuretLutzZPGV25} are used.
In contrast, as discussed above the \LTLf formulas constructed by \progmon can grow exponentially, so that each step may require exponential space (resp.~doubly-exponential time). However, our experiments suggest that this theoretical bound is rarely reached in practice.
As Rem.~\ref{rem:safety} shows, for co-safety and safety properties one satisfiability check can be omitted, which may improve the performance of \progmon in practice.

In the arithmetic case, monitoring is undecidable, and there are differences between \dfamon and \progmon w.r.t-~decidable fragments. On the one hand, while \dfamon is guaranteed to terminate for bounded lookback, MC, and IPC properties~\cite{FMPW23}, \progmon is known to terminate only for the latter two. On the other hand, \progmon terminates for all co-safety and safety properties where all iteration (resp., co-iteration) conditions are MCs or IPCs (cf.~\thmref{decidability:progressive:safety}), while there is no respective result for \dfamon.

Finally, the two approaches differ wrt. explainability: When monitoring complex properties such as our PM benchmarks that are conjunctions of numerous \Declare patterns, it may be of interest to understand which sub-property is violated or satisfied. While this is hard to achieve with the monolithic \dfamon, progression monitoring  can return the monitoring state of each conjunct~\cite{BauerFMSD2016}, and hence provide a more fine-grained response.

As future work, when monitoring multiple traces against the same specification, \progmon could be made more efficient by caching satisfiability checks and progression results \emph{across}, and not only within, traces. This would essentially amount to a lazy construction of a DFA.
Second, while we showed that progression monitoring is a decision procedure for the MC and IPC fragments, it is subject to further investigation whether this also holds for bounded lookback.


\bibliography{references}

\newpage
\appendix
\section{Appendix}

\subsection{Proofs}
\lemmaprog*
\begin{proof}
We show first the following statement for a single progression step ($\star$):
For $\psi$, $\tau$ with length $n$, and $0 \leq i < n$, $\tau,i \models \psi$ iff either $i = n-1$ and $\psi^+(\tau(i))\equiv \top$, or 
$\tau, i+1 \models \psi^+(\tau(i))$.
This is proven by induction on $\psi$, for arbitrary $i$. 
\begin{compactitem}
\item
If $\psi=\top$ or $\psi=\bot$ then $\psi^+(\tau(i))=\psi$, and the statement clearly holds.
\item
Let $\psi=p$ for $p \in P$. If $\tau,i \models p$, i.e., $p \in \tau(i)$, then $\psi^+(\tau(i))=\top$ and the statement holds.
Conversely, if $i < n-1$ and $\tau, i+1 \models \psi^+(\tau(i))$, or $\psi^+(\tau(i))\equiv\top$ then we must have $p\in \tau(i)$, so $\tau,i \models \psi$.
\item
The case of $\psi=\neg p$ for $p \in P$ is  similar: 
If $\tau,i \models \psi$ then $p \not\in\tau(i)$, so $\psi^+(\tau(i))=\top$ by definition and the statement holds.
Conversely, if either $\tau, i+1 \models \psi^+(\tau(i))$ and $i<n-1$, or $\psi^+(\tau(i))\equiv\top$ then we must have $p \not in \tau(i)$, so $\tau,i \models \psi$.
\item
For $\psi=\psi_1 \wedge \psi_2$, $\tau,i \models \psi$ holds iff $\tau,i \models \psi_1$ and $\tau,i \models \psi_2$, which is by the induction hypothesis the case iff either 
$i < n-1$ and for both $j\in \{1,2\}$ we have
$\tau, i+1 \models \psi_j^+(\tau(i))$, or $i= n-1$ and $\psi_j^+(\tau(i))\equiv\top$.
This former case is equivalent to 
$i < n-1$ and $\tau, i+1 \models \psi_1^+(\tau(i)) \wedge \psi_2^+(\tau(i)) = \psi^+(\tau(i))$, and the latter is equivalent to
 $\psi_1^+(\tau(i)) \wedge \psi_2^+(\tau(i)) = \psi^+(\tau(i))\equiv\top$ and $i = n-1$.
\item
For $\psi=\psi_1 \vee \psi_2$, $\tau,i \models \psi$ holds iff $\tau,i \models \psi_1$ or $\tau,i \models \psi_2$, which is by the induction hypothesis the case iff for some $j\in \{1,2\}$ we have
$\tau, i+1 \models \psi_j^+(\tau(i))$ and $i<n$, or $i=n$ and $\psi_j^+(\tau(i))\equiv\top$.
This former case is equivalent to 
$i < n-1$ and $\tau, i+1 \models \psi_1^+(\tau(i)) \vee \psi_2^+(\tau(i)) = \psi^+(\tau(i))$, and the latter is equivalent to
 $\psi_1^+(\tau(i)) \vee \psi_2^+(\tau(i)) = \psi^+(\tau(i))\equiv\top$ and $i = n-1$.
\item
For $\psi=\X \psi_1$, suppose first $\tau,i \models \psi$, so 
$i<n-1$ and $\tau,i+1 \models \psi_1$.
Since $\tau$ is well-formed,  $\last \not\in \tau(i)$, so $\tau,i+1 \models \psi_1$ and $\psi^+(\tau(i))=\psi_1$.
For the other direction, if $i<n-1$ and $\tau,i+1 \models \psi_1$ then by well-formedness $\last \not\in \tau(i)$, so clearly $\tau,i \models \psi$ and $\psi_1 = \psi^+(\tau(i))$.
If $i = n-1$ then by well-formedness $\last \in \tau(i)$, so $\psi^+(\tau(i))\equiv \top$ is impossible.
\item
For $\psi=\wX \psi_1$, suppose first $\tau,i \models \psi$, so either
$i<n-1$ and $\tau,i+1 \models \psi_1$, or $i=n-1$. 
In the former case, the claim holds because by well-formedness of $\tau$ we have $\last \not\in \psi^+(\tau(i))$ and $\psi^+(\tau(i)) = \psi_1$.
In the latter case, by well-formedness  $\last \in \psi^+(\tau(i))$, so by definition $\psi^+(\tau(i)) = \top$.
For the other direction, if $i<n-1$ and $\tau,i+1 \models \psi_1$ then by well-formedness $\last \not\in \tau(i)$, so clearly $\tau,i \models \psi$ and $\psi_1 = \psi^+(\tau(i))$.
If $i = n-1$ then by definition $\tau,i \models \psi$.
\item
Consider now $\psi = \psi_1 \U \psi_2$, where we have $\psi^+(\tau(i)) = $
Suppose first $\tau,i \models \psi$.
Then either (1) $\tau,i \models \psi_2$, or (2) $i < n-1$, $\tau,i \models \psi_1$, and $\tau,i+1 \models \psi$. In case (1), by the induction hypothesis either (1a) $i = n-1$ and $\psi_2^+(\tau(i)) \equiv \top$, or (1b) $i < n-1$ and $\tau,i+1 \models \psi_2^+(\tau(i))$. In  case (1a), $\last \in \tau(i)$, so $\psi^+(\tau(i)) = \psi_2^+(\tau(i))$, hence $\psi^+(\tau(i)) \equiv \top$ and the claim holds. In case (1b),  $\last \not \in \tau(b)$, so
$\psi^+(\tau(i)) = \psi_2^+(\tau(i)) \vee (\psi_1^+(\tau(i)) \wedge \psi)$. Since we have $\tau,i+1 \models \psi_2^+(\tau(i))$, also $\tau,i+1 \models \psi^+(\tau(i))$ holds.
For case (2), by the induction hypothesis and the fact that $i < n-1$, we have $\tau,i+1 \models \psi^+_1$. By well-formedness we have $\last \not \in \tau(b)$, so 
$\psi^+(\tau(i)) = \psi_2^+(\tau(i)) \vee (\psi_1^+(\tau(i)) \wedge \psi)$, and $\tau,i+1 \models \psi^+$.

For the other direction, suppose first (1) $i = n-1$ and $\psi^+(\tau(i)) \equiv \top$.
By well-formedness $\last \in \tau(i)$, so $\psi^+(\tau(i)) = \psi_2^+(\tau(i))$. By the induction hypothesis, $\tau, i \models \psi_2$, so $\tau, i\models \psi$. Otherwise (2), $i < n-1$ and $\tau, i+1 \models \psi^+(\tau(i))$, so by well-formedness 
$\psi^+(\tau(i)) = \psi_2^+(\tau(i)) \vee (\psi_1^+(\tau(i)) \wedge \psi)$. Hence either
(2a) $\tau, i+1 \models \psi_2^+(\tau(i))$ or $\tau, i+1 \models \psi_1^+(\tau(i)) \wedge \psi$.
In case (2a), by the induction hypothesis we have $\tau, i \models \psi_2$ and hence $\tau, i \models \psi$. In case (2b), both $\tau, i+1 \models \psi_1^+(\tau(i))$ and $\tau, i+1 \models\psi$, so by the induction hypothesis we have $\tau, i \models \psi_1$, and by Def.~\ref{def:semantics} also $\tau, i \models \psi$.
\item The case for $\R$ is similar.
\end{compactitem}
We now show that for $0 \leq i < n$, it holds that $\tau \models \psi$ iff either $i \geq n-1$ and $\psi^+(\tau)\equiv \top$, or $i < n-1$ and
$\tau_{> i} \models \psi^+(\tau_{\leq i})$.
The latter is equivalent to $\tau, i+1 \models \psi^+(\tau_{\leq i})$.

$(\Longrightarrow)$
By induction on $i$.
In the base case  $i=0$ and  $\tau = \trace{\tau(i)}$, so the claim follows from ($\star$).
Let $i > 0$. 
By the induction hypothesis, we get either $i-1 \geq n-1$ and $\psi^+(\tau)\equiv \top$, or $i-1 < n-1$ and
$\tau_{> i-1} \models \psi^+(\tau_{\leq i-1})$.
In the first case, if $i \geq n-1$ then also $i+1 > n-1$ and $\psi^+(\tau)\equiv \top$ holds.
Otherwise, we have $i < n-1$ and $\tau, i \models \psi^+(\tau_{\leq i-1})$, so by ($\star$) we have 
$\tau, i+1 \models \psi^+(\tau_{\leq i-1})^+(\tau(i))$, that is, $\tau_{> i} \models \psi^+(\tau_{\leq i})$.

$(\Longleftarrow)$ By induction on $i$.
In the base case, $i=0$.
If $i=n-1$ and $\psi^+(\tau_{\leq i}) = \psi^+(\tau(0))\equiv \top$, we obtain $\tau, i \models \psi$ from ($\star$).
If $i < n-1$ and $\tau_{> i} \models \psi^+(\tau(i))$, so
$\tau, i+1 \models \psi^+(\tau(i))$ then $\tau, i \models \psi$ holds by ($\star$) as well.
For $i > 0$, we again distinguish two cases.
First, let $i=n-1$ and $\psi^+(\tau_{\leq i}) \equiv \top$,
so $\psi'^+(\tau(n-1)) \equiv \top$ for $\psi' = \psi^+(\tau(0))^+\dots (\tau(i-1))= \psi^+(\tau_{\leq i-1})$. By ($\star$), $\tau, i \models \psi^+(\tau_{\leq i-1})$, and by the induction hypothesis, $\tau \models \psi$.
Second, suppose $i < n-1$ and $\tau_{> i} \models \psi^+(\tau_{\leq i})$, i.e., $\tau, i+1 \models \psi^+(\tau_{\leq i}) = \psi'^+(\tau(i))$ for
$\psi' = \psi^+(\tau(0))^+\dots (\tau(i-1))$.
By ($\star$), $\tau, i \models \psi' =  \psi^+(\tau_{\leq i-1})$, and by the induction hypothesis, $\tau \models \psi$.
\qedhere
\end{proof}

\correctness*
\begin{proof}
Let $\tau= \trace{w_0w_1 \dots w_{n-1}}$, for $n\geq 0$.
For any $i$, let $\tau_{\leq i}^\last = \trace{w_0, \dots, w_{i-2}, w_i \cup\{\last\}}$.
Let $\chi_i$ resp. $\phi_i$ be the formulas obtained in Line~\ref{line:progression:progresslast} resp. \ref{line:progression:progress} of Alg.~\ref{alg:progression:monitoring} in iteration $i$, for $0\leq i < n$.
We show by induction on $i$  that 
$\phi_{i} = \psi^+(\tau_{\leq i})$, and the procedure outputs $s_i$ such that 
$\tau_{\leq i} \models \llbracket \psi = s_i\rrbracket$.
In iteration 0, $\chi_0 = \psi^+(\trace{w_0 \cup \{\last\}})$ and 
$\phi_0 = \psi^+(\trace{w_0})$.
If $\chi_0 \equiv \top$, then by \lemref{progression} applied to $\psi$, trace $\tau_{\leq 0}^\last = \trace{w_0 \cup\{\last\}}$ and $i=0$ we have $\tau_{\leq 0}^\last \models \psi$ and hence 
$\tau_{\leq i} \models \psi$ as $\last$ does not occur in $\psi$.
So $s_0$ must be $\CS$ or $\PS$.
\begin{inparaenum}[(1)]
\item
Suppose $\neg \phi_0$ is satisfied by some trace $\tau'$, so $\tau' \not \models \phi_0$.
By \lemref{progression} applied to trace $\psi$, $\tau_{\leq{0}}\tau'$, and $i=0$, using
$(\tau_{\leq{0}}\tau')_{> 0} = \tau' \not\models \varphi_0$
we conclude $\tau_{\leq{0}}\tau' \not \models \psi$, so $s_{0}=\CS$.
\item
If $\neg \phi_0$ is unsatisfiable, so $\phi_0$ is valid, so every trace $\tau'$ satisfies $\tau' \models \phi_{0}$. 
By \lemref{progression} applied to $\psi$, $\tau_{\leq{0}}\tau'$, and $i=0$, this means that $\tau_{\leq{0}}\tau' \models \psi$, so $s_{0}=\PS$.
\end{inparaenum}
If $\chi_0 \not \equiv \top$, the reasoning is similar.

In iteration $i+1$, we assume by the induction hypothesis that $\phi_{i} = \phi^+(\tau_{\leq i})$.
Therefore, we have $\phi_{i+1} = \phi_{i}^+(\tau(i+1)) = \psi^+(\tau_{\leq {i+1}})$. For  $\tau^\last = \tau_{\leq i}\trace{ \tau(i+1) \cup \{\last\}}$, we similarly have $\chi_{i+1} = \psi^+(\tau^\last)$.
Suppose $\chi_{i+1} \equiv \top$.
By \lemref{progression} applied to $\psi$, $\tau^\last$, and $i+1$, we have that $\tau^\last \models \psi$, and hence $\tau \models \psi$, so the monitoring state is $\CS$ or $\PS$.
\begin{inparaenum}[\itshape (i)]
\item
Suppose $\neg \phi_{i+1}$ is satisfiable by some trace $\tau'$, so $\tau' \not \models \phi_{i+1}$.
By \lemref{progression} applied to trace $\psi$, $\tau_{\leq{i+1}}\tau'$, and $i+1$, using
$(\tau_{\leq{i+1}}\tau')_{> 1} = \tau' \not\models \varphi_{i+1}$
we conclude $\tau_{\leq{i+1}}\tau' \not \models \psi$, so $s_{i+1}=\CS$.
\item
If $\neg \phi_{i+1}$ is unsatisfiable, $\phi_{i+1}$ is valid, so every trace $\tau'$ satisfies $\tau' \models \phi_{i+1}$. 
By \lemref{progression} applied to $\psi$, $\tau_{\leq{i+1}}\tau'$, and $i+1$, this means that $\tau_{\leq{i+1}}\tau' \models \psi$, so $s_{i+1}=\PS$.
\end{inparaenum}
If $\chi_{i+1} \not \equiv \top$, the reasoning is similar.
\end{proof}

For the arithmetic case, we extend the equivalence relation $\equiv$ such that all formulas in the arithmetic theory without variables are evaluated, e.g. $(1 < 2) \equiv \top$ and $(3+7 = 11) \equiv \bot$.

\lemmaprogressionarith*
\begin{proof}
As for the propositional case, we show first ($\star$):
For $\psi$, $\tau$ with length $n$, and $0 \leq i < n$, $\tau,i \models \psi$ iff either $i = n-1$ and $\psi^+(\tau(i))\equiv \top$, or 
$\tau, i+1 \models \psi^+(\tau(i))$.
This is proven by induction on $\psi$, for arbitrary $i$. 
\begin{compactitem}
\item
Let $\psi$ be an arithmetic constraint $e_1 \odot e_2$. 
Suppose first that $\tau,i \models \psi$.
If $\psi$ is not well-defined, we have $i = n-1$ and $\psi$ contains $V'$. Then $\psi^+(\tau(i))\equiv \top$ by definition.
If $\psi$ is well-defined and does not contain $V'$, $\progress(\psi, \tau(i)) = \psi\tau(i) \equiv \top$, so the claim holds. If $\psi$ is well-defined and contains $V'$, $\progress(\psi, \tau(i)) = \psi\tau(i)\gamma$. Since $\tau,i \models \psi$, the assignment $\delta$ with domain $V \cup V'$ that sets $\delta(v) = \tau(i)$ if $v\in V$ and $\delta(v') = \tau(i+1)$ if $v'\in V'$ satisfies $\psi$. Hence $\tau, i+1 \models \progress(\psi, \tau(i))$.

For the other direction, suppose first that $i = n-1$ and $\psi^+(\tau(i))\equiv \top$. Then $\psi$ is either not well-defined at $i$, so $\tau,i \models \psi$, or it does not contain $V'$ In the latter case, $\progress(\psi, \tau(i)) = \psi\tau(i)\gamma \equiv \top$ implies that $\tau(i)$ satisfies $\psi$, hence $\tau,i \models \psi$.
If $i < n-1$ and $\tau, i+1 \models \psi\tau(i)\gamma$ then the assignment $\delta$ with domain $V \cup V'$ that sets $\delta(v) = \tau(i)$ if $v\in V$ and $\delta(v') = \tau(i+1)$ if $v'\in V'$ satisfies $\psi$, so $\tau,i \models \psi$.
\item All other cases are proven as in the proof of Lem.~\ref{lem:progression}.
\qedhere
\end{compactitem}
\end{proof}

\theoremprogressionarith*
\begin{proof}
This theorem is proven exactly like Thm.~\ref{thm:correctness}, replacing all invocations of Lem.~\ref{lem:progression} by Lem.~\ref{lem:progression:arith}.
\end{proof}

\decidabilityMCIPC*
 \begin{proof}
It is easy to see that for MC properties (1) also all properties in the progression space are MC properties, and 
(2) also all negations of MC properties are MC properties. The claim follows by decidability of satisfiability checking for MC proprties~\cite[Thm.~5]{GeattiGGW23}.
A similar reasoning applies to IPC properties, using \cite[Thm.~6]{GeattiGGW23}.
\end{proof}

\decidabilitysafety*
 \begin{proof}
First, note that in an input property $\psi$ satisfies the conditions of the theorem about (co)iteration conditions, then this holds also for all properties $\phi$ obtained in Line~\ref{line:progression:progress} as (co-)iteration conditions are preserved by progression.
\sarahtodo{relationship safety and progression?}

For co-safety properties, by Rem~\ref{rem:safety}, the check in Line~\ref{line:progression:sat} of Alg.~\ref{alg:progression:monitoring}  can be omitted. By the above observation, we can assume that for every formula $\phi$ obtained in Line~\ref{line:progression:progress}, all iteration conditions are MCs (or all are IPCs). For these properties, satisfiability is decidable by~\cite[Thms 5 and 6]{GeattiGGW23}.

For safety properties, by Rem~\ref{rem:safety}, the check in Line~\ref{line:progression:vio} of Alg.~\ref{alg:progression:monitoring}  can be omitted. By the above observation, we can assume that for every formula $\phi$ obtained in Line~\ref{line:progression:progress}, all co-iteration conditions are MCs (or all are IPCs). Since $\neg (\psi_1 \U \psi_2) \equiv \psi_1 \R \psi_2$, all iteration conditions of $\neg \phi$ are MCs resp. IPCs, so again the satisfiability check is decidable.
 \end{proof}

We add an observation about syntactic safety and co-safety \aLTLf properties, to illustrate that Thm.~\ref{thm:decidability:progressive:safety} can be applied in practice.

\begin{proposition}
\aLTLf formulas in NNF
\begin{compactenum}
\item are safety properties if they contain only the temporal operators $\wX$ and $\R$, and constraints with lookback are not negated; and
\item are co-safety properties if they contain only the temporal operators $\X$ and $\U$, and all constraints with lookback are negated.
\end{compactenum}
\end{proposition}
\begin{proof}[Proof (sketch)]
\begin{compactenum}
\item 
This can be seen using results by \citeauthor{Geatti2025SafetyAL} (\citeyear{Geatti2025SafetyAL}) who show that safety languages on finite words are exactly those that are prefix-closed.
The result for the propositional case is already stated there; it remains to show that arithmetic constraints with the stated property do not cause any harm. This is the case because lookback is by our definition evaluated in a weak way, with a similar semantics as $\wX$. 
\item The reasoning for the co-safety fragment is similar, the respective result for propositional \LTLf is already present in \cite{Geatti2025SafetyAL}. The reason why constraints with lookback must be negated is that this has the effect that lookback is evaluated in a ``strict'' as opposed to weak way, like the strict lookahead operator defined in~\cite{GGG22}.
\qedhere
\end{compactenum}
\end{proof}

\subsection{Experiments}

\paragraph{Propositional \LTLf.}
We first report additional details about the formulas and traces in the benchmark sets:
\begin{itemize}
\item[(SC)]
 This set contains 4471 formula/trace pairs with an average size of 5000 (maximal size 120000), average depth of 10 (maximal depth 20), and on average 50 propositions (maximally 500). All traces have length 500.
\item[(PM)] This set contains 3637 formula/trace pairs, all of the form $\G(\bigwedge_{i=1}^n\psi_i)$ where $\psi_i$ is an \LTLf version of a Declare constraint. Declare constraints are essentially patterns of temporal formulas, that can all be translated into \LTLf by a standard translation~\cite{DonadelloRMS22}. A common such pattern is e.g. $\m{Reponse}(a,b)$ for some propositional atoms $a$ and $b$, which transforms into the \LTLf formula $\G(a \to \X\F b)$. 
The formulas in the benchmark set have average size of 2000 (maximal 12000), an average depth of 9 (maximal depth 13), and on average 7 (maximally 24) propositions.
The traces have at most length 1350, on average their length is 18.
\end{itemize} 

\smallskip

\textit{Overall time required by different monitors.}
Next, we provide further data obtained from comparing the different monitoring approaches.
We comment on the overlap and differences of different techniques on (SC):
\begin{compactitem}
\item
\dfamon with Lydia and \progmon have an intersection of 2744 solved problems. Given the number of problems solved by each of the techniques, this overlap is relatively small, i.e., the approaches are quite orthogonal. This also explains why the combined technique solves so many more problems than each technique alone. The total time required by \progmon to monitor all traces in this overlap set is 2.6 times the time required by \dfamon/Lydia.
\item
\dfamon with Lydia and \dfamon with Spot have an overlap of 3432 problems. However, the Lydia monitor requires only about 54\% of the time needed by Spot for this shared set.
\item
\dfamon with Lydia and the combined monitor with Spot have an overlap of 3494 problems. \dfamon requires in total only 17\% of the time needed by the combined monitor with Spot for this set.
\item 
The overlap of \progmon and the combined monitor with Spot are 3670 problems, so 40 problems less than solved by \progmon alone. Hence, almost 99\% of the problems solved by \progmon alone are also solved by the combined setting. This illustrates that the combined monitor induces a bit of overhead, but not too much. The combined monitor requires in total 78\% of the time needed by \progmon for this overlap set.
\item
\dfamon with Lydia and the combined monitor with Lydia have an overlap of 3495 problems, i.e., almost all problems solved by \dfamon alone can also be solved in the combined setting. \dfamon requires in total only 20\% of the time needed by the combined monitor for this set.
\item 
The overlap of \progmon and the combined monitor with Lydia are 3669 problems. This combined monitor requires in total 80\% of the time needed by \progmon for this overlap set.
\end{compactitem}
\smallskip

\textit{Satisfiability checks in \progmon.}
For a trace of length $n$, \progmon performs at most $n$ satisfiability checks, on in each iteration. However, the results of satisfiability checks of formulas are cached,  and re-used in later iterations. Moreover, the obvious optimization is implemented that as soon as a permanent verdict is reached (\PS or \PV), the monitor does no more work but simply repeats the verdict until the length of the trace. Indeed, on the (SC) benchmarks, the maximum number of satisfiability checks done is 500 (all traces have length 500). However, the average number of satisfiability checks is only 61, so on average by far not every iteration requires a satisfiability check, due to caching and the fact that permanent verdicts are reached.

\textit{Comparison of \dfamon and \progmon times.}
Fig.~\ref{fig:comparison:SC2} compares the monitoring times of \dfamon (x-axis) and \progmon (y-axis) for the SC problems solved by both. The data points are separated according to the final monitoring state. 
Fig.~\ref{fig:comparison:PM} provides a similar comparison for the PM benchmarks. The fact that in both cases more points are above the $x=y$ diagonal illustrates once more that DFA monitoring is more efficient.
The figures also illustrate that the monitoring time of \progmon where the final monitoring state is $\PS$ or $\PV$
tends to be lower. In these cases, \progmon needs to show unsatisfiability of an \LTLf formula, which is harder with the tableau method than satisfiability. We suspect that times for these problems are lower because the hard problems that would take longer are simply not solvable for \progmon.

\begin{figure}
\resizebox{.48\textwidth}{!}{
\begin{tikzpicture}
\begin{axis}[
	width=.85\textwidth,
	height=.55\textwidth,
	xmode=log,
	ymode=log,
	xlabel={DFA monitoring},
	ylabel={progression monitoring},
	legend pos=north west,
    enlarge x limits = 0.03,
]
\addplot+[darkgreen,only marks, mark=x, mark size=1.2pt] table {dfa_vs_prog_SC_PS.dat};
\addlegendentry{$\PS$}

\addplot+[blue, only marks, mark=x, mark size=1.2pt] table {dfa_vs_prog_SC_CS.dat};
\addlegendentry{$\CS$}

\addplot+[orange, only marks, mark=x, mark size=1.2pt] table {dfa_vs_prog_SC_CV.dat};
\addlegendentry{$\CV$}

\addplot+[red, only marks, mark=x, mark size=1.2pt] table {dfa_vs_prog_SC_PV.dat};
\addlegendentry{$\PV$}

\addplot+[gray, dashed, domain=9e3:2e6, no markers] {x};
\addlegendentry{$y=x$}
\end{axis}
\end{tikzpicture}
}
\caption{Scatter plot comparing DFA monitoring and progression monitoring times ($\mu$s) on SC instances.}
\label{fig:comparison:SC2}
\end{figure}
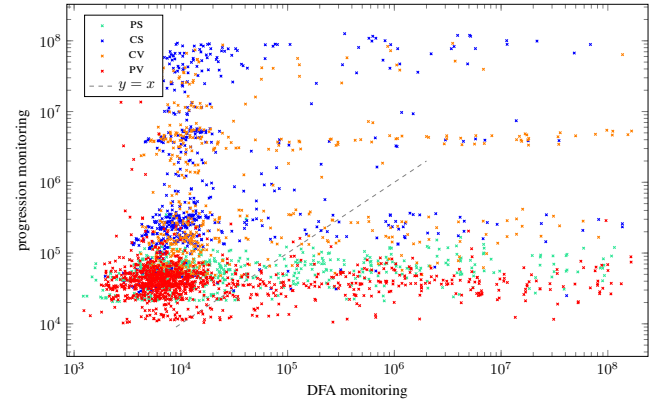
\begin{figure}
\resizebox{.48\textwidth}{!}{
\begin{tikzpicture}
\begin{axis}[
	width=.85\textwidth,
	height=.55\textwidth,
	xmode=log,
	ymode=log,
	xlabel={DFA monitoring},
	ylabel={progression monitoring},
	legend pos=north west,
    enlarge x limits = 0.03,
    enlarge y limits = 0.0
]

\addplot+[blue, only marks, mark=x, mark size=1.2pt] table {dfa_vs_prog_PM_CS.dat};
\addlegendentry{$\CS$}

\addplot+[orange, only marks, mark=x, mark size=1.2pt] table {dfa_vs_prog_PM_CV.dat};
\addlegendentry{$\CV$}

\addplot+[red, only marks, mark=x, mark size=1.2pt] table {dfa_vs_prog_PM_PV.dat};
\addlegendentry{$\PV$}

\addplot+[gray, dashed, domain=9e3:2e6, no markers] {x};
\addlegendentry{$y=x$}
\end{axis}
\end{tikzpicture}
}
\caption{Scatter plot comparing DFA monitoring and progression monitoring times ($\mu$s) on PM instances.}
\label{fig:comparison:PM}
\end{figure}
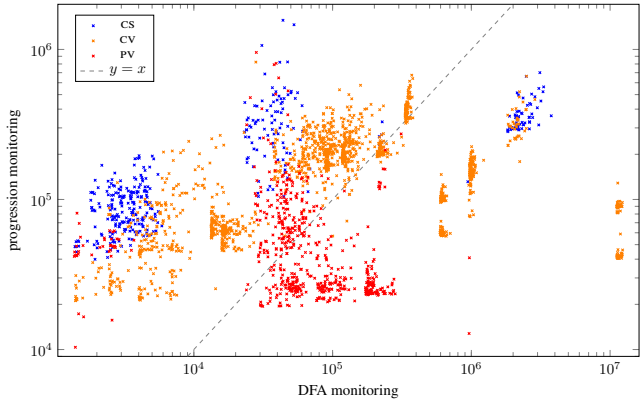

\textit{\dfamon vs. \progmon relative to formula size.}
Figs.~\ref{fig:diff:size}, \ref{fig:diff:depth}, and \ref{fig:diff:atoms} show the difference in monitoring times between \dfamon (with Lydia) and \progmon on SC problems ordered by different parameters, i.e., the y-axis shows $t_{\mathit{dfa}} - t_{\mathit{prog}}$ where $t_{\mathit{dfa}}$ is the time required by \dfamon and $t_{\mathit{prog}}$ is the time required by \progmon, both in $\mu$s. We included all formula/trace pairs where at least one approach was successful, and in case one of the two approaches timed out, $t_{\mathit{prog}}$ resp. $t_{\mathit{dfa}}$ were set to $1.8\cdot 10^8$ (3 minutes). That is, for a point above the x-axis means that \progmon was more efficient for this formula/trace pair, and for a point below the x-axis, $\dfamon$ was more efficient. In particular, dots close to the upper border to the picture mean that \dfamon timed out or needed almost 3 minutes, while \progmon took at most a few seconds. 
The x-axes show formula size (\ref{fig:diff:size}), formula depth (\ref{fig:diff:depth}), and the number of different propositional atoms in the formula (\ref{fig:diff:atoms}), respectively.  The figures illustrate that for smaller formulas the difference is often negative, i.e., \dfamon is faster than \progmon, but for larger formulas and formulas with many atoms, the difference is typically (very) positive, i.e., \progmon is (much) faster than \dfamon. 
Also, there are no large formulas (formula size above around 1000, or more than about 80 atoms) where \dfamon has an advantage over \progmon: on the contrary, it is significantly more promising to use \progmon for such inputs.
For formula depth, we could not draw any clear conclusion.

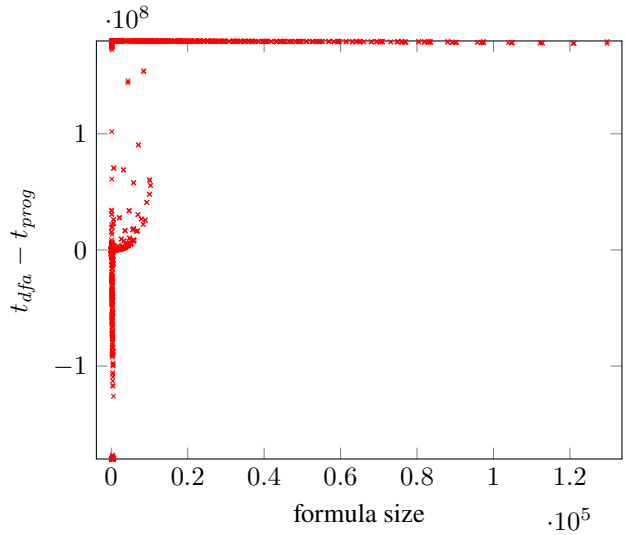
\begin{figure}
\begin{tikzpicture}
\begin{axis}[
	width=.48\textwidth,
	height=.4\textwidth,
	xmode=linear,
	ymode=linear,
	xlabel={formula size},
	ylabel={$t_{\mathit{dfa}} - t_{\mathit{prog}}$},
	legend pos=north west,
    ymin = -180000000,
    enlarge x limits = 0.03,
    enlarge y limits = 0.0
]
\addplot+[red,only marks, mark=x, mark size=1.2pt] table {adv_size.dat};
\end{axis}
\end{tikzpicture}
\caption{Monitoring time difference between \progmon and \dfamon on SC problems, ordered by formula size. It can be concluded that for formulas with size greater than about 1000, progression monitoring performs better.}
\label{fig:diff:size}
\end{figure}

\begin{figure}
\begin{tikzpicture}
\begin{axis}[
	width=.48\textwidth,
	height=.4\textwidth,
	xmode=linear,
	ymode=linear,
	xlabel={formula depth},
	ylabel={$t_{\mathit{dfa}} - t_{\mathit{prog}}$},
	legend pos=north west,
    ymin = -180000000,
    enlarge x limits = 0.03,
    enlarge y limits = 0.0
]
\addplot+[blue,only marks, mark=x, mark size=1.2pt] table {adv_depth.dat};
\end{axis}
\end{tikzpicture}
\caption{Monitoring time difference between \progmon and \dfamon on SC problems, ordered by formula depth.}
\label{fig:diff:depth}
\end{figure}
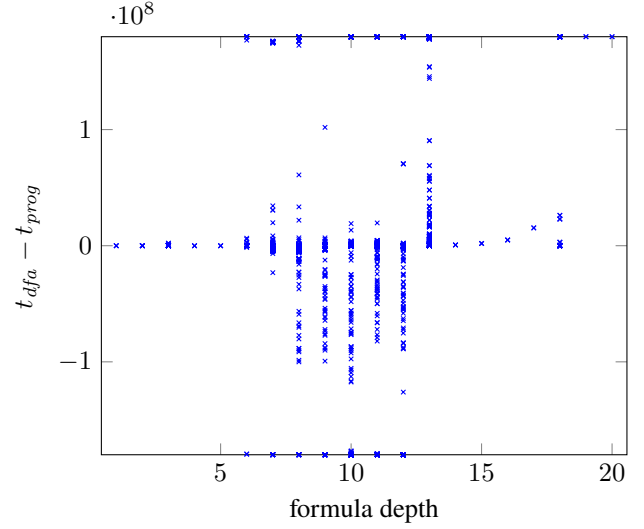

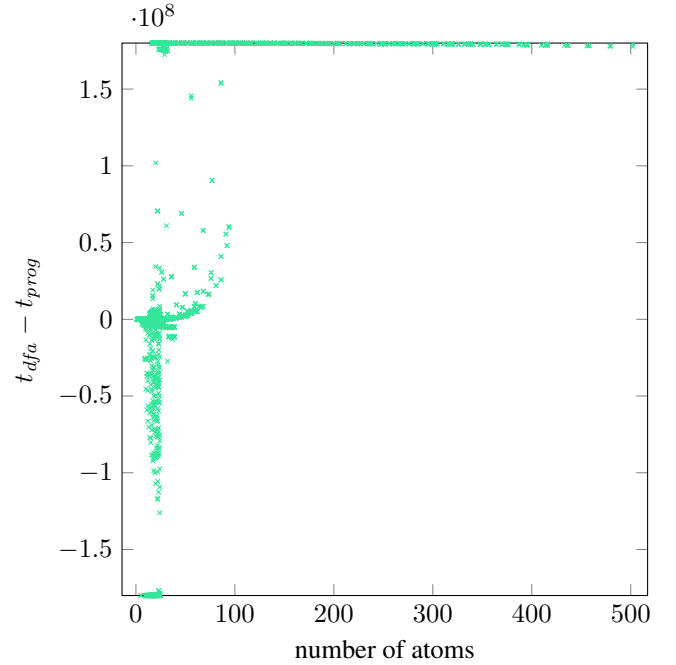
\begin{figure}
\begin{tikzpicture}
\begin{axis}[
	width=.48\textwidth,
	height=.5\textwidth,
	xmode=linear,
	ymode=linear,
	xlabel={number of atoms},
	ylabel={$t_{\mathit{dfa}} - t_{\mathit{prog}}$},
	legend pos=north west,
    ymin = -180000000,
    enlarge x limits = 0.03,
    enlarge y limits = 0.0
]
\addplot+[darkgreen,only marks, mark=x, mark size=1.2pt] table {adv_atoms.dat};
\end{axis}
\end{tikzpicture}
\caption{Monitoring time difference between \progmon and \dfamon on SC problems, ordered by the number of different atoms in formulas. It can be concluded that for formulas with more than about 80 atoms, progression monitoring performs better.}
\label{fig:diff:atoms}
\end{figure}

\textit{Combined monitoring.}
We also implemented and tested a combined monitoring procedure that uses Lydia for formula construction, rather than Spot. The results are very similar. The version with Spot is about 5\% faster, solving about the same number of problems on SC, but a few more on PM, cf. Fig.~\ref{fig:cactus2} (which is as Fig.~\ref{fig:cactus} but with two combined approaches).

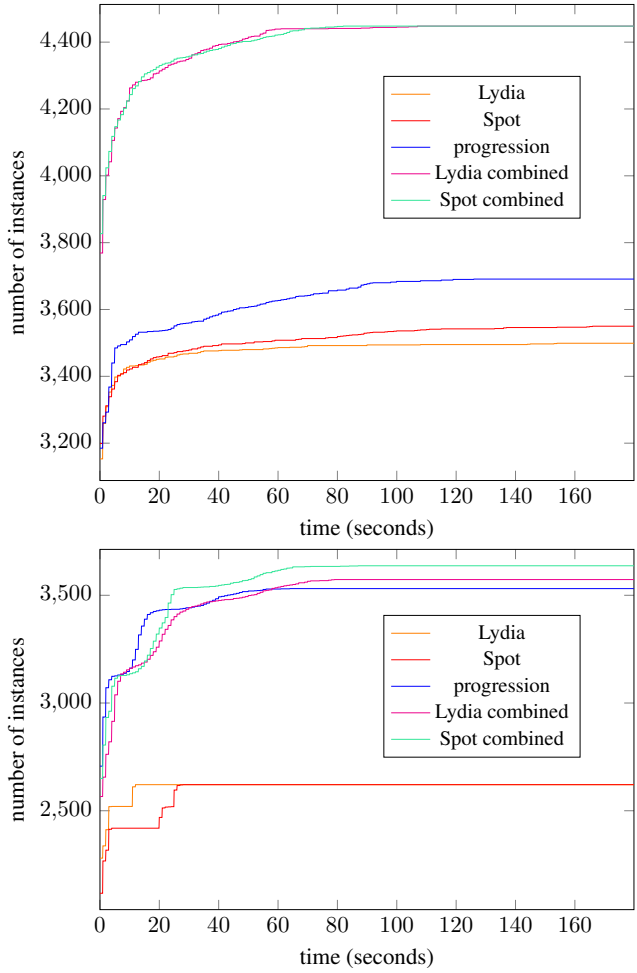
\begin{figure}
\resizebox{.48\textwidth}{!}{
\begin{tikzpicture}
\begin{axis}[
	width=.55\textwidth,
	height=.5\textwidth,
	xlabel={time (seconds)},
	ylabel={number of instances},
	legend style={at={(0.9,0.7)},anchor=east, font=\footnotesize},
    enlarge x limits = 0.0,
    enlarge y limits = 0.05,
]
\addplot+[orange, no marks] table {DFA_cactus_SC.dat};
\addlegendentry{Lydia}

\addplot+[red, no marks] table {spot_cactus_SC.dat};
\addlegendentry{Spot}

\addplot+[blue, no marks] table {prog_cactus_SC.dat};
\addlegendentry{progression}

\addplot+[magenta, no marks] table {combined_cactus_SC.dat};
\addlegendentry{Lydia combined}

\addplot+[darkgreen, no marks] table {spot_combined_cactus_SC.dat};
\addlegendentry{Spot combined}
\end{axis}
\end{tikzpicture}
}
\\
\resizebox{.48\textwidth}{!}{
\begin{tikzpicture}
\begin{axis}[
	width=.55\textwidth,
	height=.4\textwidth,
	xlabel={time (seconds)},
	ylabel={number of instances},
	legend style={at={(0.9,0.62)},anchor=east, font=\footnotesize},
    enlarge x limits = 0.0,
    enlarge y limits = 0.05,
]
\addplot+[orange, no marks] table {DFA_cactus_PM.dat};
\addlegendentry{Lydia}

\addplot+[red, no marks] table {spot_cactus_PM.dat};
\addlegendentry{Spot}

\addplot+[blue, no marks] table {prog_cactus_PM.dat};
\addlegendentry{progression}

\addplot+[magenta, no marks] table {combined_cactus_PM.dat};
\addlegendentry{Lydia combined}

\addplot+[darkgreen, no marks] table {spot_combined_cactus_PM.dat};
\addlegendentry{Spot combined}

\end{axis}
\end{tikzpicture}
}
\caption{Number of instances solved within time limit for the SC (above) and PM benchmarks (below), including different combined approaches.}
\label{fig:cactus2}
\end{figure}

\paragraph{Arithmetic \LTLf.}
Tab.~\ref{tab:arithmetic} shows the detailed results on the arithmetic benchmark set.
One clearly sees that \progmon can construct monitors for many more problems, and is more efficient.

\newcommand{\timeout}{$\infty$}
\begin{table}
\begin{footnotesize}
\begin{tabular}{@{}l@{\:}llrrr@{}}
\hline
&formula &  & \textsc{PM} & \textsc{CG} & ada \\
\hline
(A) &blood\_sugar/1 & int & 0.114 & 1.269 & 0.539 \\
(A) &blood\_sugar/1 & real & 0.098 & 1.247 & 0.566 \\
(A) &heart\_rate/1 & int & 0.091 & 1.211 & 0.567 \\
(A) &heart\_rate/1 & real & 0.079 & 1.159 & 0.520 \\
(A) &temperature/1 & int & 0.060 & 0.224 & 0.491 \\
(A) &temperature/1 & real & 0.065 & 0.210 & 0.494 \\
(A) &temperature/2 & int & 0.074 & 0.457 & 0.502 \\
(A) &temperature/2 & real & 0.060 & 0.447 & 0.505 \\
(A) & monotone\_inc & int & 0.037 & \timeout & 6.039 \\
(A) &monotone\_inc & real & 0.045 & 0.075 & 6.205 \\
(B) &LIA\_scalable1\_10 & int & 3.390 & \timeout & \timeout \\
(B) &LIA\_scalable1\_10 & real & 3.359 & \timeout & \timeout \\
(B) &LIA\_scalable1\_100 & int & 6.934 & \timeout & \timeout \\
(B) &LIA\_scalable1\_100 & real & 6.906 & \timeout & \timeout \\
(B) &LIA\_scalable1\_500 & int & \timeout & \timeout & \timeout \\
(B) &LIA\_scalable1\_500 & real & \timeout & \timeout & \timeout \\
(B) &LIA\_scalable2\_10 & int & 0.534 & 2.519 & \timeout \\
(B) &LIA\_scalable2\_10 & real & 0.607 & 2.483 & \timeout \\
(B) &LIA\_scalable2\_100 & int & \timeout & \timeout & \timeout \\
(B) &LIA\_scalable2\_100 & real & \timeout & \timeout & \timeout \\
(B) &LIA\_scalable2\_20 & int & 1.324 & 18.298 & \timeout \\
(B) &LIA\_scalable2\_20 & real & 1.332 & 17.076 & \timeout \\
(B) &LIA\_scalable2\_5 & int & 0.276 & 2.570 & \timeout \\
(B) &LIA\_scalable2\_5 & real & 0.284 & 2.369 & \timeout \\
(C) &emission/1 & int & 0.092 & \timeout & \timeout \\
(C) &emission/1 & real & 28.649 & \timeout & \timeout \\
(C) &emission/1 & int & 0.091 & \timeout & \timeout \\
(C) &emission/1 & real & 27.272 & \timeout & \timeout \\
(C) &emission/2 & int & 0.057 & \timeout & \timeout \\
(C) &emission/2 & real & 0.058 & \timeout & \timeout \\
(C) &emission/2 & int & 0.056 & \timeout & \timeout \\
(C) &emission/2 & real & 0.056 & \timeout & \timeout \\
(C) &vac\_cleaner & int & 0.029 & 0.054 & \timeout \\
(C) &vac\_cleaner & real & 0.031 & 0.052 & \timeout \\
(C) &vac\_cleaner & int & 0.045 & 0.054 & \timeout \\
(C) &vac\_cleaner & real & 0.040 & 0.058 & \timeout \\
\end{tabular}
\end{footnotesize}
\caption{Experiments for arithmetic \LTLf}
\label{tab:arithmetic}
\end{table}

\end{document}